%% file: EACQ-Singleton.17mar23-arXiv.tex
\documentclass[aps,nopacs,nokeys,superscriptaddress,11pt,twoside,notitlepage,a4paper]{revtex4-1}

\usepackage{graphicx,epic,eepic,epsfig,amsmath,latexsym,amssymb,verbatim,color,xcolor,mathtools}

\usepackage{enumerate}
\usepackage{tikz}
\usepackage{floatrow}

\usepackage{comment}

\usepackage{orcidlink}

\usepackage{theorem}
\newtheorem{definition}{Definition}

\newtheorem{lemma}[definition]{Lemma}

\newtheorem{theorem}[definition]{Theorem}
\newtheorem{corollary}[definition]{Corollary}

\newcommand{\qed}{\hfill$\square$}

\newenvironment{proof}{\noindent \textbf{{Proof~} }}{\qed}
\newenvironment{remark}{\noindent \textbf{{Remark~}}}{}

\mathchardef\ordinarycolon\mathcode`\:
\mathcode`\:=\string"8000
\def\vcentcolon{\mathrel{\mathop\ordinarycolon}}
\begingroup \catcode`\:=\active
  \lowercase{\endgroup
  \let :\vcentcolon
  }

\newcommand{\nc}{\newcommand}
\nc{\rnc}{\renewcommand}
\nc{\lbar}[1]{\overline{#1}}
\nc{\bra}[1]{\langle#1|}
\nc{\ket}[1]{|#1\rangle}
\nc{\ketbra}[2]{|#1\rangle\!\langle#2|}
\nc{\braket}[2]{\langle#1|#2\rangle}
\nc{\proj}[1]{| #1\rangle\!\langle #1 |}
\nc{\avg}[1]{\langle#1\rangle}
\nc{\Rank}{\operatorname{Rank}}
\nc{\smfrac}[2]{\mbox{$\frac{#1}{#2}$}}
\nc{\tr}{\operatorname{Tr}}
\nc{\ox}{\otimes}
\nc{\dg}{\dagger}
\nc{\dn}{\downarrow}
\nc{\cA}{{\cal A}}
\nc{\cB}{{\cal B}}
\nc{\cC}{{\cal C}}
\nc{\cD}{{\cal D}}
\nc{\cE}{{\cal E}}
\nc{\cF}{{\cal F}}
\nc{\cG}{{\cal G}}
\nc{\cH}{{\cal H}}
\nc{\cI}{{\cal I}}
\nc{\cJ}{{\cal J}}
\nc{\cK}{{\cal K}}
\nc{\cL}{{\cal L}}
\nc{\cM}{{\cal M}}
\nc{\cN}{{\cal N}}
\nc{\cO}{{\cal O}}
\nc{\cP}{{\cal P}}
\nc{\cQ}{{\cal Q}}
\nc{\cR}{{\cal R}}
\nc{\cS}{{\cal S}}
\nc{\cT}{{\cal T}}
\nc{\cU}{{\cal U}}
\nc{\cV}{{\cal V}}
\nc{\cW}{{\cal W}}
\nc{\cX}{{\cal X}}
\nc{\cY}{{\cal Y}}
\nc{\cZ}{{\cal Z}}
\nc{\csupp}{{\operatorname{csupp}}}
\nc{\qsupp}{{\operatorname{qsupp}}}
\nc{\var}{{\operatorname{var}}}
\nc{\Var}{{\operatorname{Var}}}
\nc{\rar}{\rightarrow}
\nc{\lrar}{\longrightarrow}
\nc{\polylog}{{\operatorname{polylog}}}
\nc{\1}{{\openone}}
\nc{\wt}{{\operatorname{wt}}}

\def\g{\gamma}
\def\d{\delta}
\def\e{\epsilon}

\def\r{\rho}
\def\s{\sigma}

\def\ph{\varphi}

\def\o{\omega}

\def\G{\Gamma}

\def\Ph{\Phi}

\def\O{\Omega}

\nc{\RR}{{{\mathbb R}}}
\nc{\CC}{{{\mathbb C}}}
\nc{\FF}{{{\mathbb F}}}
\nc{\NN}{{{\mathbb N}}}
\nc{\ZZ}{{{\mathbb Z}}}
\nc{\PP}{{{\mathbb P}}}
\nc{\QQ}{{{\mathbb Q}}}
\nc{\UU}{{{\mathbb U}}}
\nc{\EE}{{{\mathbb E}}}
\nc{\id}{{\operatorname{id}}}

\nc{\CHSH}{{\operatorname{CHSH}}}

\nc{\<}{\langle}
\rnc{\>}{\rangle}
\nc{\Hom}[2]{\mbox{Hom}(\CC^{#1},\CC^{#2})}
\nc{\rU}{\mbox{U}}

\nc{\ob}[1]{#1}

\nc{\SEP}{{\text{SEP}}}
\nc{\NS}{{\text{NS}}}
\nc{\LOCC}{{\text{LOCC}}}
\nc{\PPT}{{\text{PPT}}}
\nc{\EXT}{{\text{EXT}}}
\nc{\Sym}{{\operatorname{Sym}}}

\nc{\ERLO}{{E_{\text{r,LO}}}}
\nc{\ERLOCC}{{E_{\text{r,LOCC}}}}
\nc{\ERPPT}{{E_{\text{r,PPT}}}}
\nc{\ERLOCCinfty}{{E^{\infty}_{\text{r,LOCC}}}}
\nc{\Aram}{{\operatorname{\sf A}}}

\newcommand{\bunderline}[1]{\underline{#1\mkern-4mu}\mkern4mu }

\nc{\baro}{\overline{\o}}
\nc{\barO}{\overline{\O}}
\nc{\barg}{\overline{\g}}
\nc{\barG}{\overline{\G}}
\nc{\Mhat}{\widehat{M}}
\nc{\Xhat}{\widehat{X}}
\nc{\Phhat}{\widehat{\Ph}}
\nc{\uX}{\bunderline{X}}
\nc{\uA}{\bunderline{A}}
\nc{\uB}{\bunderline{B}}

\nc{\tdC}{\tilde{C}}
\nc{\tdQ}{\tilde{Q}}
\nc{\tdE}{\tilde{E}}

\usepackage[bb=boondox]{mathalfa}

\usepackage[normalem]{ulem}
\usepackage{cancel}
\newcommand{\strikered}[1]{\ifmmode\textcolor{red}{\cancel{#1}}\else\textcolor{red}{\sout{#1}}\fi}

\begin{document}

\title{Singleton Bounds for Entanglement-Assisted \protect\\ Classical and Quantum Error Correcting Codes\footnote{A short version of this work has been presented at ISIT 2022 and is included in its proceedings \cite{EACQ-Singleton-ISIT}.}}

\author{Manideep Mamindlapally\,\orcidlink{0000-0002-8157-3972}}
\email{manideepyx@iitkgp.ac.in}
\affiliation{IIT Kharagpur, Kharagpur, India - 721302}

\author{Andreas Winter\,\orcidlink{0000-0001-6344-4870}}
\email{andreas.winter@uab.cat}
\affiliation{ICREA \&{} Grup d'Informaci\'o Qu\`antica, Departament de F\'isica, Universitat Aut\`onoma de Barcelona, 08193 Bellaterra (BCN), Spain}
\affiliation{Institute for Advanced Study, Technische Universit\"at M\"unchen, Lichtenbergstra{\ss}e 2a, D-85748 Garching, Germany}

\begin{abstract}
We show that entirely quantum Shannon theoretic methods, based on von Neumann entropies and their properties, can be used to derive Singleton bounds on the performance of entanglement-assisted hybrid classical-quantum (EACQ) error correcting codes. 
Concretely, we show that the triple-rate region of qubits, cbits and ebits of possible EACQ codes over arbitrary alphabet sizes is contained in the quantum Shannon theoretic rate region of an associated memoryless erasure channel, which turns out to be a polytope. 
We show that a large part of this region is attainable by certain EACQ codes, whenever the local alphabet size (i.e.~Hilbert space dimension) is large enough, in keeping with known facts about classical and quantum maximum distance separable (MDS) codes: in particular, all of its extreme points and all but one of its extremal lines. The attainability of the remaining one extremal line segment is left as an open question.
\end{abstract}

\date{17 March 2023}

\maketitle

\section{Introduction}
Quantum error correcting codes (QECC) are subject to various universal constraints relating block length, alphabet size, minimum distance and code rate, very much like classical error correcting codes. In particular, the classical Singleton bound \cite{Singleton} has a satisfying quantum version for QECC \cite{q-Singleton,q-Singleton-Rains,KlappeneckerSarvepalli:subsystem-Singleton}, which has been extended to entanglement-assisted and catalytic QECC (EAQECC and CQECC) \cite{ea-q-Singleton,catalytic-q-Singleton}. Indeed, since their proposal, EAQECC have enjoyed considerable attention from coding theorists, both in the block coding and the convolutional coding setting \cite{Galindo-et-al,LaiBrun:imperfect,Nadkarni,Allahmadi-et-al,Chen:ea,LaiAshikhmin,WildeBrunBabar:EA-turbo}.
Note that the original EAQECC Singleton bound was found to be erroneous in general \cite{Grassl-counterex}, which was put right in the recent paper \cite{EAQECC-Singleton}.
Although these bounds all have different forms, they are united in that they express the ability of the code to correct erasure errors. Classically, the Singleton bound is attained for MDS codes, which exist for sufficiently large alphabet size. Similarly, attaining these quantum Singleton bounds defines suitable quantum MDS (QMDS) codes and their entanglement-assisted generalisations. 

The present paper grew out of an attempt to understand better the results of \cite{EAQECC-Singleton}, where the most complete quantum Singleton bound so far was derived for EAQECC, in the form of a two-dimensional convex region in the ebit-qubit plane, into which all possible codes necessarily fall (as a function of block length and minimum distance). Investigating the tightness of the bound exhibited two types of codes attaining the boundary of the allowed region, one a genuinely entanglement-assisted quantum code dubbed EAQ, the other a classical MDS code piggy-backed onto a simple teleportation protocol. This suggested that, to obtain a full understanding of the codes involved, one should extend the investigation to hybrid classical and quantum codes, assisted by entanglement (EACQ codes), which we provide here. As our main result, we prove Singleton bounds for such EACQ codes, in the form of a convex triple-tradeoff region in the ebit-cbit-qubit space (as a function of block length and minimum distance). 

Hybrid classical and quantum error correcting codes have been considered in several papers before, mostly without entanglement-assistance \cite{GrasslLuZeng:hybrid,NemecKlappi:hybrid,NemecKlappi:hybrid-paper,Cao-et-al:hybrid,NemecKlappi:hybrid-detecting}; the EACQ codes as considered by us have been introduced in \cite{KremskyHsiehBrun}, although we allow additionally catalysis (recycling) of the three basic resources, making our bounds more general. Classical and quantum hybrid error correcting codes have been generalised in \cite{beny2007generalization,beny2007PRA,Majidy:hybrid,NemecKlappi:subsystem-hybrid} to classical and quantum hybrid subsystem codes in an operator algebraic setting. We stress that our bounds apply to either kind of code.
Importantly for our approach, the triple tradeoff between ebits, cbits and qubits has been treated repeatedly in the Shannon-theoretic setting of a given channel (often i.i.d. on the block of $n$ physical systems) and small errors. A precursor was the breakthrough paper by Devetak and Shor \cite{devetak2005capacity}, which showed how to analyse the capacity region of joint classical and quantum information transmission over a given noisy channel. 
For us, the paper by Hsieh and Wilde \cite{HsiehWilde} is fundamental, which derives a multi-letter capacity formula for the triple tradeoff, of which we take the converse proof and develop it in several directions. We follow essentially the very developed, rigorous exposition of Wilde \cite[Ch.~25]{wildebeast}.

\medskip
{\bf Results.} 
We give here an overview of our main results, which also serves as a guide to the paper. In Section \ref{sec:problem} we review the definition of EACQ codes and pose the problem of characterising all triples of catalytic rates attainable for given block length and minimum distance, and then discuss preliminaries in Section \ref{sec:prelim}. 
After that: 
\begin{itemize}
    \item We state Hsieh-Wilde's converse theorem (\cite[Thm.~1]{HsiehWilde}) in Section \ref{sec:info-theory-converse} and give a complete proof from first principles, and generalised both to arbitrary (one-shot) channels and the catalytic setting, in the Appendix \ref{app:converse}, the latter having previously been accomplished in \cite[Ch.~25]{wildebeast}.

  \item We use this general converse to derive the triple-tradeoff rate region for the i.i.d.~erasure channel, correcting a gap in \cite{HsiehWilde}, where the erasure probability $\delta$ was assumed to be less than $\frac12$, in Section \ref{sec:triple-tradeoff-erasure}. (Wilde \cite[Ch.~25.5.3]{wildebeast} is much more complete, but proves additivity only for $\delta\leq\frac12$.)

  \item We then use the general one-shot converse for a channel that randomly erases $d-1$ of the $n$ physical systems and with error probability set to $0$, to derive the Singleton bound for EACQ codes, subsuming both classical Singleton bounds and all previously known quantum Singleton bounds; the obtained region is the same as for an i.i.d. erasure channel with erasure probability $\delta=\frac{d-1}{n}$, in Section \ref{sec:triple-tradeoff-Singleton}.

  \item We analyse the geometric shape of the EACQ Singleton region, determining its extreme points and extremal lines; we can then show that large parts of the boundary are indeed attained, whenever the alphabet size is large enough, in Section \ref{sec:attainability}. One line segment remains to be shown to be attainable, to prove our entire region to be optimal, which we leave as an open question, and which we discuss among other things in the concluding Section \ref{sec:conclusion}.
\end{itemize}

\section{Problem Setting}
\label{sec:problem}
Following \cite[Ch.~25]{wildebeast}, we begin by defining the task of hybrid classical and quantum communication via a noisy channel, i.e.~a linear completely positive and trace preserving (cptp) map $\cN:\cL(A)\rightarrow\cL(B)$, assisted by entanglement, in the one-shot setting and allowing for a certain (small) decoding error. Here, $A$ and $B$ are complex Hilbert spaces associated to the quantum systems of sender and receiver, which for convenience we will throughout assume to be of finite dimension; $\cL(A)$ is the space of all linear operators (matrices) on $A$, and likewise for $B$.  

\begin{figure}[ht]
    \centering
    \includegraphics[width=\textwidth]{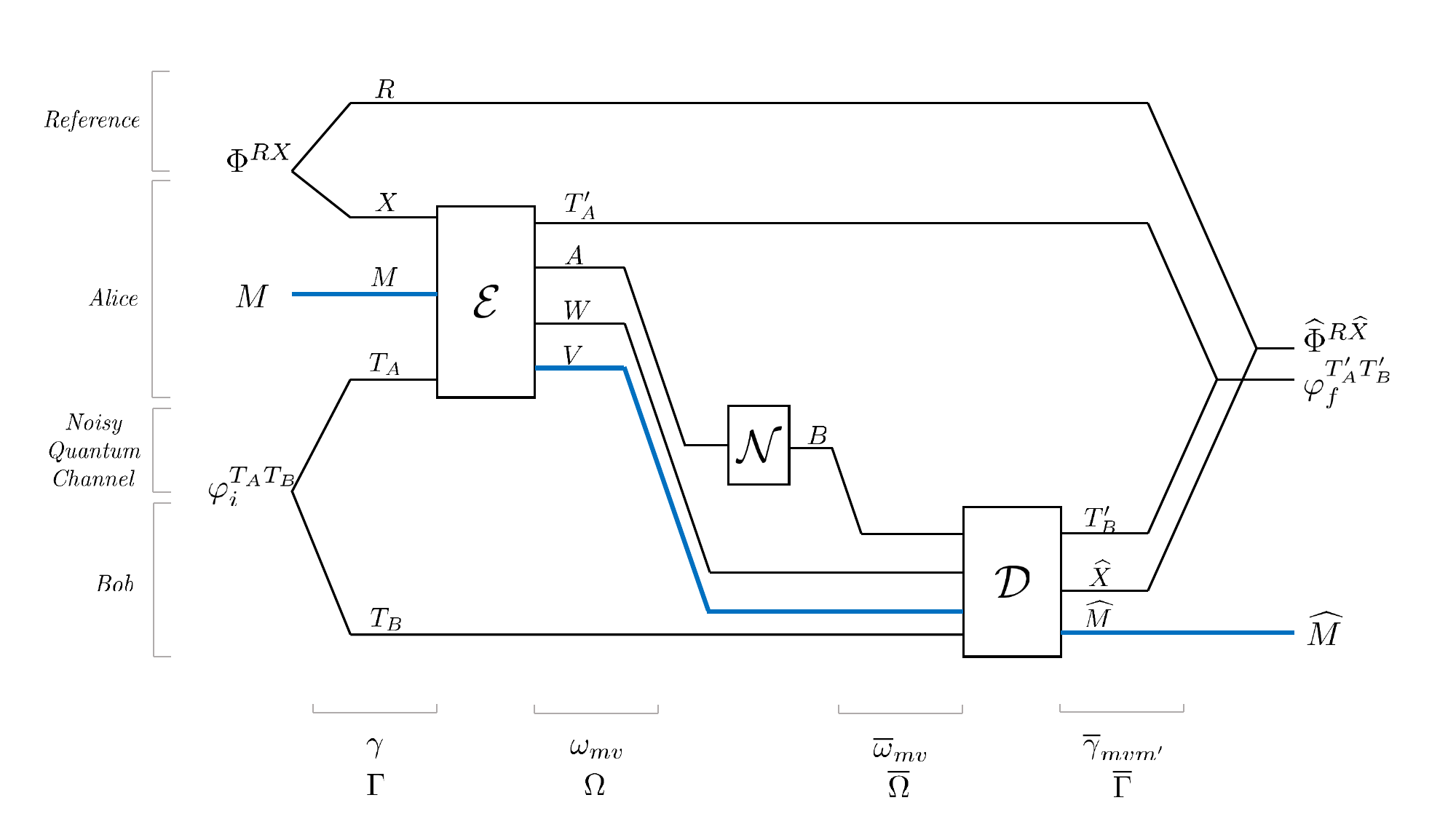}
    \caption{Communication diagramme of the most general catalytic entanglement-assisted classical and quantum code. We discuss four snapshots of the involved systems at different times, at which we wish to analyse their joint state: at starting time (initial state $\Gamma$), after the encoding ($\Omega$), after the action of the noisy channel ($\overline{\Omega}$), and after the decoding ($\overline{\Gamma}$). 
    The starting quantum systems $X$ and $R$ constitute $Q_2$ qubits each, the message $M$ has $C_2$ cbits, and the entangled systems $T_A$ and $T_B$ carry an equivalent of $E_1$ pure entangled ebits.
    The intermediate system $W$ carries $Q_1$ qubits and $V$ carries $C_1$ cbits. The input system $A$ is passed through a noisy channel $\cN$ to produce output system $B$.
    On decoding, the output quantum system $\hat{X}$ contains $Q_2$ qubits, the output classical message $\hat{M}$ has $C_2$ cbits. Additionally, a system $T_B'$ is generated that has $E_2$ ebits of entanglement with Alice's system $T_A'$.}
    \label{fig:problem_setup}
\end{figure}

We consider a general entanglement-assisted classical and quantum communication setup as in Fig.~\ref{fig:problem_setup}. We have communicating agents, a sender \emph{Alice} and a receiver \emph{Bob}. Apart from the channel $\cN$, they start with (a) a shared entangled quantum state $\ph_i^{T_A T_B}$ of $E_1=S(\ph_i^{T_A})$ \emph{ebits}, (b) a quantum channel capable of transmitting $Q_1$ \emph{qubits}, (c) a classical channel capable of transmitting $C_1$ \emph{cbits}. 
Using encoding and decoding operations, they wish to achieve (a) transfer of $Q_2$ \emph{qubits} of quantum information, (b) transfer of $C_2$ \emph{cbits} of classical information, and (c) regeneration of $E_2=S(\ph_f^{T_A'})$ \emph{ebits} of shared entanglement in the form of a shared state $\ph_f^{T_A'T_B'}$. This results in a net transfer of $Q = Q_2 -Q_1$ \emph{qubits}, $C = C_2 - C_1$ \emph{cbits} on using up of $E = E_1 - E_2$ \emph{ebits}. The entire diagram, defined by the maps $\cD$ and $\cE$ is called an \emph{entanglement-assisted classical and quantum error correcting code (EACQ code) of error $\epsilon$ for the channel $\cN$}. We speak simply of an EACQ code if the error is implied to be $\epsilon=0$. Furthermore, for an $n$-partite input system $A^n=A_1\ldots A_n$ composed of $A_j$ that are all $q$-dimensional Hilbert spaces, we say that an EACQ code has minimum distance $d$ if it has error $\epsilon=0$ for the following \emph{block erasure channel}, which uniformly randomly erases $d-1$ out of $n$ of the system $A_i$:
\begin{equation}
  \label{eq:block-erasure-channel-def}
  \cE_{q,d-1,n}(\rho) = \frac{1}{{n\choose d-1}} \sum_{J\subset[n],\atop |J|=d-1} \left(\tr_{A_J}\rho\right)^{B_{[n]\setminus J}} \ox \proj{\perp}^{B_J},
\end{equation}
where $B_j=A_j\oplus\CC\ket{\perp}$ and $\ket{\perp}^{B_J} = \bigotimes_{j\in J} \ket{\perp}^{B_j}$. 
We adopt the taxonomy of \cite{KremskyHsiehBrun} for an EACQ code of block length $n$ and minimum distance $d$: it is denoted $[\![n,k:c,d;e]\!]_q$, where $k=\frac{Q}{\log q}$, $c=\frac{C}{\log q}$ and $e=\frac{E}{\log q}$. 

\medskip
The objective now is to chararacterise the permissible values of $C$, $Q$ and $E$ for a given block length $n$ such that a code with small or vanishing error exists. To formalize this, we analyse the states describing all involved systems at four different times and then define an associated error term. In the beginning, the initial states are prepared:
\begin{align}
    \g^{RXT_AT_B} &= \Ph^{RX} \ox \ph_i^{T_AT_B}, \notag\\
    \G^{MRXT_AT_B} &= \sum_{m\in\cM} \frac{1}{|\cM|} \ket{m}\bra{m}^{M} \ox \g^{RXT_AT_B}. \\
    \intertext{On encoding using a collection of CPTP maps $\left( \cE_{v|m} ^{ XT_A \rightarrow T_A' A W } \right)_{v \in \cV}$ on classical input $m\in\cM$, we have:}
    p_{V|M}(v|m) &= \tr \cE_{v|m} ^{ XT_A \rightarrow T_A' A W } \left(\g^{RXT_AT_B}\right), \notag\\
    \o_{mv}^{RT_A'AWT_B} &= \frac{1}{p_{V|M}(v|m)} \cE_{v|m} ^{ XT_A \rightarrow T_A' A W } \left(\g^{RXT_AT_B}\right), \notag\\
    \O^{MVRT_A'AWT_B} &= \sum_{ \substack{m \in \cM \\ v \in \cV}} \frac{1}{|\cM|} p_{V|M}(v|m) \ket{mv}\bra{mv}^{MV} \ox \o_{mv}^{RT_A'AWT_B}. \label{def:O}\\
    \intertext{After the action of the channel:}
    \baro_{mv}^{RT_A'BWT_B} &= \cN^{A\rightarrow B} \left( \o_{mv}^{RT_A'AWT_B} \right), \notag\\
    \barO^{MVRT_A'BWT_B} &= \sum_{\substack{m \in \cM \\ v \in \cV}} \frac{1}{|\cM|} p_{V|M}(v|m) \ket{mv}\bra{mv}^{MV} \ox \baro_{mv}^{RT_A'BWT_B}. \label{def:barO}  %
    \intertext{And finally, after decoding using a collection of CPTP maps $\left( \cD_{m'|v} ^{ BWT_B\rightarrow \Xhat T_B' }\right)_{m'\in\cM}$ on classical input $v \in \cV$ to the decoder:}
    p_{\Mhat|MV} (m'|mv) &= \tr \cD_{m'|v} ^{ BWT_B\rightarrow \Xhat T_B' } \left(\baro_{mv}^{RT_A'BWT_B}\right), \notag\\
    \barg_{mvm'}^{R\Xhat T_A'T_B'}&= \frac{1}{p_{\Mhat|MV} (m'|mv)}\cD_{m'|v} ^{ BWT_B\rightarrow \Xhat T_B' } \left(\baro_{mv}^{RT_A'BWT_B}\right), \notag\\
    \barG^{MV\Mhat R \Xhat T_A' T_B' } &= \sum_{ \substack{m \in \cM \\ v \in \cV \\ m' \in \cM}} \frac{1}{|\cM|} p_{V|M}(v|m) p_{\Mhat|MV}(m'|mv) \ket{mvm'}\bra{mvm'}^{MV\Mhat} \ox \barg_{mvm'}^{R\Xhat T_A'T_B'}.
\end{align}

The ideal output state pertaining to the quantum information is the maximally entangled state $\Phi^{R\Xhat} = \id^{X \rightarrow \Xhat} \left( \Phi^{RX} \right)$, the ideal final entanglement state is arbitrary, but has to be a pure state $\ph_f^{T_A'T_B'}$, and the ideal state pertaining to the classical information is the perfectly correlated classical state $\overline{\Phi}^{M\widehat{M}} = \frac{1}{|\cM|}\sum_m \proj{m}^M \ox \proj{m}^{\widehat{M}}$. The ideal output is the product state
\begin{align}
    \barG_{\text{ideal}}^{M \Mhat R \Xhat T_A' T_B'} &= \overline{\Phi}^{M\widehat{M}} \ox \Phi^{R\Xhat} \ox \ph_f^{T_A'T_B'}. \label{eq:def:barGideal}
\end{align}
We say that the protocol (the code) has error $\epsilon$, if
\begin{equation}
  \frac12 \left\| \barG^{M \Mhat R \Xhat T_A' T_B'} - \barG_{\text{ideal}}^{M \Mhat R \Xhat T_A' T_B'} \right\|_1 \leq \epsilon, \label{eq:distancebound}
\end{equation}
where $\|X\|_1 = \tr |X| = \tr\sqrt{X^\dagger X}$ denotes the trace norm (aka Schatten-$1$-norm).
Note that this implies, by the contractivity of the trace norm under partial trace and more generally cptp maps, that
\begin{align*}
  \Pr\left\{M\neq \widehat{M}\right\} 
    = \frac12 \left\|\barG^{M \Mhat}-\overline{\Phi}^{M\widehat{M}}\right\|_1 
    &\leq \epsilon, \\
  \frac12 \left\| \barG^{R\Xhat} - \Phi^{R\Xhat} \right\|_1 
    &\leq \epsilon, \\
  \frac12 \left\| \barG^{T_A'T_B'} - \ph_f^{T_A' T_B'} \right\|_1 
    &\leq \epsilon. 
\end{align*}

As discussed at the start of this section, the other important parameters of the code are the initial entanglement $E_1=S\left(\ph_i^{T_A}\right)$ (where $S$ denotes the von Neumann entropy), the invested qubit transmission $Q_1=\log|W|$ (with $|W|$ denoting the dimension of the Hilbert space $W$, for both see the next section), and the invested cbit transmission $C_1=\log|V|$; furthermore the final generated entanglement $E_2=S\left(\ph_f^{T_A'}\right)$, the effected qubit transmission $Q_2=\log|X|$, and the effected cbit transmission $C_2=\log|M|$.
In the sequel we will derive bounds on the net (also called \emph{amortized}) rates $E=E_1-E_2$, $Q=Q_2-Q_1$ and $C=C_2-C_1$. 
Note the different treatment of the net entanglement rate, in two ways: first, it is defined in terms of entropies rather than a logarithmic dimension, so as not to restrict to maximal entanglement as initial or final state (but if $\varphi_i$ or $\varphi_f$ is maximally entangled, then $E_1=\log|T_A|$, $E_2=\log|T_A'|$, respectively); secondly, by prior convention in the treatment of the present problem, $E$ is the net rate of entanglement \emph{consumption}, whereas $C$ and $Q$ are net rates of resource \emph{production}.

\section{Preliminaries}
\label{sec:prelim}
Here we collect definitions and facts known from prior work that will be used in the later sections. As we have already done in the previous section introducing our problem setting, quantum systems are denoted by capital letters, which we use also, without danger of confusion, to denote the underlying complex Hilbert space. As a rule, our Hilbert spaces, $A$, $B$, etc, are finite-dimensional throughout the paper, the dimension of $A$, i.e.~the cardinality of any basis of $A$, being denoted $|A|$. We use the same notation $|\cM|$ for the cardinality of a finite set $\cM$, justified by the information-theoretic parallelism that while an alphabet of size $n$ can encode $\log n$ classical bits, a Hilbert space of dimension $d$ can encode $\log d$ qubits. The logarithm $\log$ is by default the binary logarithm, unless explicitly specified. 

With that, the von Neumann entropy of a state $\rho$ on a system $A$ is defined as $S(A)_\rho = S(\rho) = -\tr \rho\log\rho$, and the conditional entropy of a state $\rho$ on a biparite system $AB$ as $S(A|B)_\rho = S(AB)_\rho-S(B)_\rho$, which also equals the negative coherent information $I(A\rangle B)_\rho := -S(A|B)_\rho$. The quantum entropy is subject to a number of fundamental relations, principal among them strong subadditivity $I(A:C|B)_\rho = S(AB)_\rho + S(BC)_\rho - S(B)_\rho - S(ABC)_\rho$ for any tripartite state \cite{LiebRuskai:SSA}, and the equivalent weak monotonicity property $S(A|B)_\rho + S(A|C)_\rho \geq 0$, both for an arbitrary tripartite state $\rho$ on $ABC$. Furthermore, the following two uniform continuity bounds.

\begin{lemma}[Fannes inequality \cite{Fannes,Audenaert,AW:S-continuity}]
\label{lemma:Fannes}
For any two states $\rho$ and $\sigma$ on a system $A$ with $\frac12\| \r^A - \s^A \|_1 \leq \epsilon \leq 1$, 
\begin{equation*} 
     \left| S(A)_\r - S(A)_\s \right| \leq \e \log |A| + H_2(\e),
\end{equation*}
where $H_2(\epsilon)=-\e\log\e-(1-\e)\log(1-\e)$ is the binary entropy. 
\hfill\qed
\end{lemma}

\begin{lemma}[Alicki-Fannes inequality \cite{AlickiFannes,AW:S-continuity}]
\label{lemma:AF}
For any two states $\rho$ and $\sigma$ on a composite system $AB$ with $\frac12\| \r^{AB} - \s^{AB} \|_1 \leq \epsilon \leq 1$, 
\begin{equation*}
  \left| S(A|B)_\r - S(A|B)_\s \right| 
    = \left| I(A\rangle B)_\r - I(A\rangle B)_\s \right| 
    \leq 2\e \log |A| + g(\e),
\end{equation*}
where $g(x) = (1+x)H_2\left(\frac{x}{1+x}\right) = (x+1)\log(x+1)- x\log x$. 
Note that $g$ is a monotonically increasing, concave function with $g(x)\geq H_2(x)$.
\hfill\qed 
\end{lemma}

\begin{lemma}[Cf.~{\protect\cite[Lemmas ~2~and~3]{EAQECC-Singleton}}]
  \label{lemma:average-entropies}
  Consider the $(n+1)$-party system $A^nZ = A_1A_2\ldots A_nZ$, and for a subset $I \subset [n]$ of the ground set denote $A_I = \bigotimes_{i\in I} A_i$. Let $1\leq\mu\leq m\leq n$. Then, with respect to any state $\rho$ on $A^nZ$,
  \begin{align}
    \label{eq:average-plain}
    \overline{s}_m := 
    \frac1m \EE_{I\subset[n],\atop |I|=m} S(A_I)_\rho
      &\leq \frac1\mu \EE_{J\subset[n],\atop |J|=\mu} S(A_J)_\rho
      =: \overline{s}_\mu, \\
    \label{eq:average-conditional}
    \frac1m \EE_{I\subset[n],\atop |I|=m} S(A_I|Z)_\rho
      &\leq \frac1\mu \EE_{J\subset[n],\atop |J|=\mu} S(A_J|Z)_\rho.
  \end{align}
  where the expectation values in both bounds are with respect to uniformly random subsets $I,J \subset [n]$ of the ground set, of cardinality $|I|=m$, $|J|=\mu$, respectively.
  \hfill\qed
\end{lemma}

The following Lemma comes from \cite{EAQECC-Singleton}, but its proof is buried inside the proof of Theorem 6 of that paper, and uses a different notation, so we reproduce it here in full. 

\begin{lemma}[Cf.~{\protect\cite[Proof~of~Theorem~6]{EAQECC-Singleton}}]
  \label{lemma:crazy-inequality}
  Consider an $n$-party system $A^n$ of $n$ $q$-dimensional Hilbert spaces $A_1A_2\ldots A_n$. For a subset $I \subset [n]$ of the ground set denote $A_I = \bigotimes_{i\in I} A_i$. Then, with respect to any state $\rho$ on $A^n$, and any $d>\frac{n}{2}+1$, 
  \begin{equation}
  \label{eq:crazy-inequality}
    n\overline{s}_n = S(A^n) 
     \leq \EE_{|J|=n-d+1}S(A_J) + (d-1)t = (n-d+1)\overline{s}_{n-d+1} + (d-1)t,
  \end{equation}
  where $t$ is defined by the relation
  \[ 
    (n-2d+2)t 
      := (n-d+1)\overline{s}_{n-d+1}-(d-1)\overline{s}_{d-1} 
      = \EE_{|J|=n-d+1}S(A_J) - \EE_{|I|=d-1}S(A_I).
  \]
  It satisfies $0\leq t\leq \overline{s}_{d-1} \leq \log q$.
\end{lemma}

\begin{proof}
We start with the bounds on $t$. To prove its non-negativity, observing $n-2d+2<0$, we have to show $\EE_{|J|=n-d+1}S(A_J) \leq \EE_{|I|=d-1}S(A_I)$. To this end, consider uniformly random subsets $I$, $J$ and $J'$ such that $|I|=d-1$, $|J|=|J'|=n-d+1$, jointly distributed in such a way that $I\stackrel{.}{\cup}J=[n]$ and $J'\subset I$. Let $K=:I\setminus J'$ and $I':=K\cup J$, which allows us to note that $K$ and $I'$ are uniformly random subsets such that $|K|=2d-2-n$ and $|I'|=d-1$. Now, 
\[\begin{split}
  2\left(\EE_{|I|=d-1}S(A_I) - \EE_{|J|=n-d+1}S(A_J)\right)
    &= \left(\EE_{|I|=d-1}S(A_I) - \EE_{|J'|=n-d+1}S(A_{J'})\right) \\
    &\phantom{==}
       + \left(\EE_{|I'|=d-1}S(A_{I'}) - \EE_{|J|=n-d+1}S(A_J)\right) \\
    &= \EE \left( S(A_K|A_{J'}) + S(A_K|A_J) \right) 
     \geq 0,
\end{split}\]
where the first line holds because $I$ and $I'$ have the same distribution, and likewise $J$ and $J'$; the second line is by definition of the conditional von Neumann entropy; the final inequality follows from strong subadditivity (in the form of weak monotonicity) \cite{LiebRuskai:SSA} for each term of the expectation.

To prove $t\leq \log q$, we start from Lemma~\ref{lemma:average-entropies}, Equation (\ref{eq:average-plain}), which tells us $\overline{s}_{d-1}\leq \overline{s}_{n-d+1}$, or equivalently $(n-d+1)\overline{s}_{d-1}\leq (n-d+1)\overline{s}_{n-d+1}$. Thus, subtracting $(d-1)\overline{s}_{d-1}$ from both sides, by definition of $t$ we get $(n-2d+2)\overline{s}_{d-1}\leq (n-2d+2)t$, yielding the desired $t\leq \overline{s}_{d-1}\leq\log q$. 

\medskip
Finally, we address the inequality (\ref{eq:crazy-inequality}): consider a fixed $J\subset[n]$ with $|J|=n-d+1$, and $I:=[n]\setminus J$ with $|I|=d-1$. With a uniformly random $K\subset I$ with $|K|=2d-2-n$, Lemma \ref{lemma:average-entropies}, Equation~(\ref{eq:average-conditional}) tells us
\[
  S(A^n)-S(A_J) = S(A_I|A_J) 
                \leq \frac{d-1}{2d-2-n}\EE_{|K|=2d-2-n} S(A_K|A_J). 
\]
Taking the average over $J$ (and hence $I$) as well, and noting that $I':=K\cup J$ is a uniformly random subset with $|I'|=d-1$, we arrive at 
\[\begin{split}
  S(A^n) 
    &\leq \EE_{|J|=n-d+1} S(A_J) + \frac{d-1}{2d-2-n} \EE_{|J|=n-d+1\atop |K|=2d-2-n} S(A_K|A_J) \\
    &= \EE_{|J|=n-d+1} S(A_J) + \frac{d-1}{2d-2-n} \left( \EE_{|I'|=d-1} S(A_{I'})- \EE_{|J|=n-d+1} S(A_J)\right) \\
    &= (n-d+1)\overline{s}_{n-d+1} + (d-1)t,
\end{split}\]
the last line by definition.
\end{proof}

\section{Interlude: information theoretic converse}
\label{sec:info-theory-converse}
Here we re-derive the converse part of \cite[Thm.~1]{HsiehWilde}, for the amortized (net) rates $E$, $Q$ and $C$, and using somewhat more standard arguments compared to the proof in the cited paper, as indeed shown in \cite[Ch.~25]{wildebeast}.

\begin{theorem}[Capacity region one-shot converse bound]
\label{thm:converse}
For an EACQ error correcting code with error $\epsilon$ that uses $E_1$ \emph{ebits}, $Q_1$ \emph{qubits}, $C_1$ \emph{cbits} to generate $E_2$ \emph{ebits} and transmits $Q_2$ \emph{qubits}, $C_2$ \emph{cbits} over a quantum channel $\cN$, there exists a quantum state 
\begin{equation}
  \label{eq:sigma}
  \s^{UAB} = \sum_u p(u) \proj{u}^U \ox \cN^{A'\rightarrow B} \left(\ph_{u}^{AA'}\right),
\end{equation}
where the $p(u)\geq 0$ are probabilities and the $\ph_u^{AA'}$ are pure states with $|A|=|A'|$, such that 
\begin{align}
  \label{eq:converse:C+2Q}
  C + 2Q    &\leq I(UA:B)_\s + 2\epsilon (C_2+Q_2) + g(\epsilon), \\
  \label{eq:converse:Q-E}
  Q - E     &\leq I(A\rangle BU)_\s + 2\epsilon(Q_2+\log|T_A'|) + g(\epsilon), \\
  \label{eq:converse:C+Q-E}
  C + Q - E &\leq I(U:B)_\s + I(A\rangle BU)_\s + 2\epsilon(C_2+Q_2+\log|T_A'|) + 2g(\epsilon),
\end{align} 
holds for the net communication resource productions $Q = Q_2 - Q_1$ and $C = C_2 - C_1$, and net entanglement consumption $E = E_1 - E_2$.
\end{theorem}
The full proof of the theorem is reproduced in Appendix~\ref{app:converse}.

\medskip
The terms containing the error in Theorem \ref{thm:converse}, Eqs.~(\ref{eq:converse:C+2Q}), (\ref{eq:converse:Q-E}) and (\ref{eq:converse:C+Q-E}), vanish in the limit $\epsilon\rightarrow 0$, and indeed for $\epsilon=0$, so we introduce a notation for the information theoretic (weak converse) rate region:

\begin{equation}\begin{split}
  \label{eq:rate-region-general}
  \cR^{(1)}(\cN) 
    &:= \Bigl\{ (C,Q,E)\in\RR^3 \text{ s.t. }\exists\; \sigma^{UAB} \text{ as in Equation~(\ref{eq:sigma}) with } C + 2Q \leq I(UA:B)_\s, \Bigr.\\
    &\phantom{========}
     \Bigl. Q - E  \leq I(A\rangle BU)_\s,\,
            C + Q - E \leq I(U:B)_\s + I(A\rangle BU)_\s \Bigr\}.
\end{split}\end{equation}

\medskip
\begin{remark}
We stress that the region described in Theorem \ref{thm:converse} and in Equation~(\ref{eq:rate-region-general}) is only a converse bound, inasmuch it is not necessarily tight. We call it one-shot, because it does not require any product or other structure of the channel. 
In fact, Wakakuwa, Nakata and collaborators \cite{Wakakuwa-1,Wakakuwa-2,Wakakuwa-3} have derived one-shot achievability and converse bounds using the more familiar smooth min-entropies, which in general have to be considered tighter. We do not use them because of the difficulty of evaluating the min-entropy expressions in general and in particular in the case of erasure channels that is of interest to us here. As a side note, the achievability bounds of \cite{Wakakuwa-1,Wakakuwa-2,Wakakuwa-3} always carry positive error terms, due to the random coding technique, which make them unsuitable for the present objective of error correcting codes (i.e.~zero error).
On the other hand, by Hsieh and Wilde \cite{HsiehWilde,wildebeast} the regularisation of the region $\cR^{(1)}(\cN)$ is asymptotically achieved for product channels. And, as we shall see below, for certain channels with appropriate structure, such as the block erasure channel, which latter is indeed permutation covariant. 
\end{remark}

\section{The triple-tradeoff capacity region for the erasure channel}
\label{sec:triple-tradeoff-erasure}
Now we specialise Theorem \ref{thm:converse} to the case of an i.i.d.~tensor power of an erasure channel, single-letterising the bound in the process. This example had been discussed in \cite{HsiehWilde}, but using an ad-hoc argument rather than reduction to the general converse bound. In \cite[Thm.~25.5.3]{wildebeast}, Wilde does just that, but his technique still only gives the asymptotic capacity region for erasure probability $\delta\leq\frac12$. 
Here we redo the argument, simplifying the previous ones, and extending the result to arbitrary erasure probabilities $\delta\in[0,1]$ by reducing the analysis to block erasure channels whose capacity region we prove in Section \ref{sec:triple-tradeoff-Singleton}. To account for the asymptotic behaviour of the information quantities $C$, $Q$ and $E$, we introduce associated rates by letting $C=n\tdC$, $Q=n\tdQ$ and $E=n\tdE$.
The main challenge is to get a single-letterised capacity region in terms of $\tdC$, $\tdQ$ and $\tdE$. 

\begin{theorem}[Capacity region for i.i.d.~erasure channel]
\label{thm:iid-erasure}
For an i.i.d.~erasure channel $\cE_{q,\delta}^{\ox n}$ with probability of erasure $\d$, in the limit of $n\rightarrow\infty$ and $\epsilon\rightarrow 0$, the system of converse inequalities from Theorem \ref{thm:converse} for an EACQ code of net cbit rate $\tdC$, net qubit rate $n\tdQ$ and net ebit rate $\tdE$, respectively, reduces to the region of triples $(\tdC, \tdQ, \tdE)$ such that there exists a $t \in [0,\log q]$ with
\begin{align}
  \label{eq:erasure:C+2Q}
  \tdC + 2\tdQ &\leq (1-\d) (\log q + t), \\
  \label{eq:erasure:Q-E}
  \tdQ - \tdE  &\leq (1-2\d)t,            \\
  \label{eq:erasure:C+Q-E}
  \tdC+\tdQ-\tdE  &\leq (1-\d) \log q - \d t.
\end{align}
\end{theorem}

\medskip
Note that this region is also attainable, by the general coding theorem (cf.~\cite[Thm.~25.5.3]{wildebeast}). We discuss it below in Section \ref{sec:attainability} for the present concrete case, which is much simpler. 
\medskip

\begin{proof}
In the limit $n\rightarrow\infty$ and $\epsilon\rightarrow 0$, for the case of a general channel $\cN$ we are talking about 
\[
  (\tdC, \tdQ, \tdE)\in\cR(\cN) 
    := \overline{\bigcup_{n\geq 1}\frac1n \cR^{(1)}\left(\cN^{\ox n}\right)},
\]
the regularisation of the single-letter region $\cR(\cN)$.

We shall first calculate the single-letter region $\cR^{(1)}(\cE_{q,\delta})$, by simply plugging in $\epsilon=0$ and the erasure channel $\cE_{q,\delta}$ into Theorem \ref{thm:converse}. 
For this, consider a quantum state
\begin{equation}
    \s_0^{UAA'} = \sum_{u \in \cU} p(u) \ketbra{u}{u}^U \ox \ph_u^{AA'},
\end{equation}
which when passed through the erasure channel $\cE_{q,\delta}$ becomes
\begin{align}
    \s^{UAB} &= \sum_{u \in \cU} p(u) \ketbra{u}{u}^U \ox \left( (1-\d)\ph_u^{AB} + \d \ph_u^A \ox \ketbra{\perp}{\perp}^B \right)\\
             &= (1-\d) \s_1^{UAB} + \d \s _{\perp}^{UAB},
\end{align}
where $\s_1$ and  $\s_{\perp}$ are defined as 
\begin{align}
    \s_1 ^{UAB} &= \text{id}^{A' \to B} (\s_0^{UAA'}) = \sum_{u\in \cU} p(u) \ketbra {u}{u}^U \ox \ph_u^{AB}, \\
    \s _{\perp}^{UAB} &= \left( \sum_{u\in\cU} p(u) \ketbra{u}{u}^U \ox \ph_u^{A} \right) \ox \ketbra{\perp}{\perp}^B.
\end{align}
Note that $\ph_u^{AB}$ is a pure state for every $u \in \cU$, but $\ph_u^A = \tr_B \ph_u^{AB}$ may not be.
Evaluating the information quantities in Theorem \ref{thm:converse} for $\sigma ^{UAB}$, we have
\begin{align}
    I(AU:B)_\s &= I(U:B)_\s + I(A:B|U)_\s \notag\\
               &= S(U)_\s - S(U|B)_\s + S(A|U)_s - S(A|UB)_\s \label{eq:iiderasure:0a}\\
               &\begin{multlined}
                = (1-\d) S(U)_{\s_1} + \d S(U)_{\s_\perp} + H_2(\d) \\
               - (1-\d) S(U|B)_{\s_1} - \d S(U|B)_{\s_\perp} - H_2(\d) \\
               + (1-\d) S(A|U)_{\s_1} + \d S(A|U)_{\s_\perp} + H_2(\d) \\
               - (1-\d) S(A|UB)_{\s_1} - \d S(A|UB)_{\s_\perp} - H_2(\d) \label{eq:iiderasure:0b}
               \end{multlined}\\
               &=(1-\d) I(U:A')_{\s_0} + (1-\d) I(A:A'|U)_{\s_0} \label{eq:iiderasure:0c}\\
               &= (1-\d) ( S(A')_{\s_0} + S(A'|U)_{\s_0}), \label{eq:iiderasure:1}\\
    I(A\rangle B U)_\s &= S(B|U)_\s  - S(AB|U)_\s \notag\\
                       & \begin{multlined}
                           = (1-\d) S(B|U)_{\s_1} + \d S(B|U)_{\s_\perp} + H_2(\d) \\
                           - (1-\d) S(AB|U)_{\s_1} - \d S(AB|U)_{\s_\perp} - H_2(\d)\label{eq:iiderasure:1a}
                       \end{multlined}\\
                       &= (1-\d) S(B|U)_{\s_1} - \d S(A|U)_{\s_\perp} -\d S(B|U)_{\s _{\perp}}\label{eq:iiderasure:1b}\\
                       &= (1-\d) S(A'|U)_{\s_0} - \d S(A|U)_{\s_0} \label{eq:iiderasure:1c}\\
                       &= (1-2\d) S(A'|U)_{\s_0}, \label{eq:iiderasure:2}\\
    I(U:B)_\s + I(A\rangle B U)_\s  
               &= (1-\d)I(U:A')_{\s_0} + (1-2\d)S(A'|U)_{\s_0} \notag\\
               &=(1-\d)S(A')_{\s_0} -\d S(A'|U)_{\s_0}, \label{eq:iiderasure:3}
\end{align}
where $H_2(\d) = \d \log \d + (1-\d) \log (1-\d)$ is the binary entropy function. Equation~\eqref{eq:iiderasure:0a} follows from the definition of quantum mutual information. We get Equation~\eqref{eq:iiderasure:0b} expanding each of the entropy expression in terms of $\d$ by the chain rule. 
All the $H_2(\d)$ terms cancel each other. Furthermore, $S(U)_{\s_\perp} = S(U|B)_{\s_\perp}$, $S(A|U)_{\s_\perp} = S(A|UB)_{\s_\perp}$ because of product states $\sigma_\perp^{UAB} = \s_\perp^{AU} \ox \s_\perp^{B}$.
Applying the definition of mutual information on the remaining terms yields Equation~\eqref{eq:iiderasure:0c} which simplifies to Equation~\eqref{eq:iiderasure:1}. 
In a similar fashion, Equation~\eqref{eq:iiderasure:1a} comes from expanding the entropy expression in terms of $\d$. The $H_2(\d)$ terms cancel out, and 
$S(AB|U)_{\s_1}=0$ because $\ph_u^{AB}$ is a pure state and $S(B|U)_{\s_\perp}=0$ because $\sigma_\perp^{UAB} = \s_\perp^{UA} \ox \s_{\perp}^{B}$; this yields Equation~\eqref{eq:iiderasure:1b}. 
Because of the identity map from $\s_0$ to $\s_1$ and using the fact that $\s_\perp ^B = \ketbra {\perp}{\perp}$ is a pure state and that $\s_\perp ^{AU} = \s_1 ^{AU}$, we get Equation~\eqref{eq:iiderasure:1c}. This further reduced to Equation~\eqref{eq:iiderasure:2} because of the pure state $\ph_u ^{AA'}$. We get Equation~\eqref {eq:iiderasure:3} by simply expanding the mutual information term.
We know that 
\begin{equation}
    0 \leq S(A'|U)_{\s_0} \leq S(A')_{\s_0} \leq \log q. \label{eq:iiderasure:4}
\end{equation}
By defining  $t:=S(A'|U)_{\s_0}$ and combining the Inequalities \eqref{eq:iiderasure:1}, \eqref{eq:iiderasure:2}, \eqref{eq:iiderasure:3} and \eqref{eq:iiderasure:4} we obtain the claimed form for $\cR^{(1)}(\cE_{q,\delta})$.

It remains to show the additivity of the region when considering tensor powers, i.e.~$\cR^{(1)}\left(\cE_{q,\delta}^{\ox n}\right) = n\cR^{(1)}(\cE_{q,\delta})$. In \cite[Thm.~25.5.3]{wildebeast} this is done for $\delta\leq\frac12$ using a duality approach and exploiting the degradability of the erasure channel in the said regime; for $\delta>\frac12$ the additivity had been an open question. 
Here, we present a different proof, reducing a code for the i.i.d.~erasure channel $\cE_{q,\d}^{\ox n}$ to one for a suitable block erasure channel as in Section \ref{sec:problem}, and using the converse bounds from Theorem \ref{thm:block-erasure} in Section \ref{sec:triple-tradeoff-Singleton} below, which works for all $\delta$. For an error weight $0\leq w \leq n$, recall the definition of the block erasure channel $\cE_{q,w,n} : A^n \rightarrow B^n$,
\begin{equation*}
  \cE_{q,w,n}(\rho) 
  = \frac{1}{{n\choose w}} \sum_{J\subset[n],\atop |J|=w} \left(\tr_{A_J}\rho\right)^{B_{[n]\setminus J}} \ox \proj{\perp}^{B_J},
\end{equation*}
where $\ket{\perp}^{B_J} = \bigotimes_{j\in J} \ket{\perp}^{B_j}$. Then it can be checked immediately that 
\[
  \cE_{q,\d}^{\ox n}
   = \sum_{v=0}^n {n\choose v}\delta^v(1-\delta)^{n-v}\cE_{q,v,n},
\]
and furthermore that for every $w\leq v$, there exists a cptp map $\cD_{v|w}:B^n\rightarrow B^n$ such that $\cE_{t,v,n} = \cD_{v|w}\circ\cE_{t,w,n}$. This map is easily described: it measures which $w$ of the $n$ system $B^n$ are erased, and takes a uniformly random subset of the $v-w$ non-erased systems to erase them as well. 
Choose a $\lambda<\delta$ and let $w:=\lfloor \lambda n \rfloor$, so that by Hoeffding's bound,
\[
  \eta = \sum_{v=0}^{w-1} {n\choose v}\delta^v(1-\delta)^{n-v} \leq e^{-n(\delta-\lambda)^2}.
\]
Define now the channel (Note the normalization with $1-\eta$)
\[\begin{split}
  \cE_{q,\d}^{(n)} 
    &:= \frac{1}{1-\eta} \sum_{v=w}^n {n\choose v}\delta^v(1-\delta)^{n-v}\cE_{q,v,n} \\
    &=   \sum_{v=w}^n \frac{1}{1-\eta} {n\choose v}\delta^v(1-\delta)^{n-v}\cD_{v|w}\circ\cE_{q,w,n} \\
    &=:  \cD'\circ\cE_{q,w,n},
\end{split}\]
which in other words is a degraded version of $\cE_{q,w,n}$; at the same time, by its definition it satisfies $\frac12\|\cE_{q,\d}^{(n)} - \cE_{q,\d}^{\ox n}\|_\diamond \leq \eta$. Thus, a code for $\cE_{q,\d}^{\ox n}$ with error $\epsilon$ is one for $\cE_{q,\d}^{(n)}$ with error $\epsilon+\eta$, and by using the post-processing $\cD'$ before the decoding, also for $\cE_{q,w,n}$ with the same error and the same rates. Thus, the converse Theorem \ref{thm:block-erasure} for $\cE_{q,w,n}$ applies, keeping in mind its error $\epsilon+\eta\rightarrow 0$ as $n\rightarrow\infty$, showing that there exists a $t\in[0;\log q]$ with 
\begin{align*}
  \tdC+2\tdQ  &\leq (1-\lambda)(\log q + t), \\
  \tdQ-\tdE   &\leq (1-2\lambda)t, \\
  \tdC + \tdQ - \tdE &\leq (1-\lambda)\log q - \lambda t.
\end{align*}
As this is true for all $\lambda<\delta$, the claim follows. 
\end{proof}

\medskip
Using Fourier-Motzkin elimination of $t$ one can rewrite the region of Theorem \ref{thm:iid-erasure} in terms of linear inequalities for $\tdC$, $\tdQ$ and $\tdE$ only, as follows. 

\begin{theorem}[Capacity region for i.i.d.~erasure channel, alternate form]
\label{thm:iid-erasure-eliminated}
For an i.i.d.~erasure channel $\cE_{q,\delta}^{\ox n}$ with probability of erasure $\d$, in the limit of $n\rightarrow\infty$ and $\epsilon\rightarrow 0$, the converse bounds from Theorem \ref{thm:iid-erasure} for an EACQ code of net $n\tdC$ cbits, $n\tdQ$ qubits, $n\tdE$ ebits, respectively, can be expressed as follows.
Namely $(\tdC, \tdQ, \tdE)\in \cR(\cE_{q,\d}) = \cR^{(1)}(\cE_{q,\d})$ if and only if
\begin{align}
    \tdC + 2\tdQ           &\leq 2(1-\delta)\log q, \label{eq:iid-erasure-eliminated:thm:1}\\
  \tdQ -\tdE               &\leq \max\{0,1-2\delta\}\log q, \label{eq:iid-erasure-eliminated:thm:2}\\
  \tdC + \tdQ - \tdE             &\leq (1-\delta)\log q, \label{eq:iid-erasure-eliminated:thm:3}\\
  \tdC+(1+\d)\tdQ-(1-\d)\tdE &\leq (1-\delta)\log q, \label{eq:iid-erasure-eliminated:thm:4}\\
  &\hspace{-4.55cm}
  \begin{dcases}
    \phantom{\frac{3\d-1}{1-\d}}
    \frac{1-2\d}{1-\d}\tdC +\tdQ -\tdE &\!\!\!\!\leq (1-2\d)\log q\ \text{ if }\ \delta\leq\frac12, \\
    \frac{2\d-1}{1-\d}\tdC +\frac{3\d-1}{1-\d}\tdQ-\tdE &\!\!\!\!\leq (2\d-1)\log q\ \text{ if }\ \delta\geq\frac12.
  \end{dcases}\label{eq:iid-erasure-eliminated:thm:5}
\end{align}
\end{theorem}
The proof, which is obtained by Fourier-Motzkin elimination of $t$ from the bound in Theorem \ref{thm:iid-erasure}, is found in Appendix \ref{app:iid-erasure-eliminated}. 

\medskip
However, for the subsequent analysis the form of inequalities in Theorem \ref{thm:iid-erasure-eliminated} is not particularly useful. Instead, we use Theorem \ref{thm:iid-erasure} to understand the geometry of the region defined. Namely, for fixed $t\in[0;\log q]$, notice that
\[\begin{split}
  \cS_t := \bigl\{(\tdC, \tdQ, \tdE)\in\RR^3 \,:\, 
              &\tdC + 2\tdQ \leq (1-\d) (\log q + t), \bigr. \\
           \bigl.& \tdQ - \tdE  \leq (1-2\d)t,\,
              \tdC + \tdQ -\tdE  \leq (1-\d) \log q - \d t \bigr\}
\end{split}\]
is a simplicial cone, being defined by three linearly independent inequalities. The region identified in Theorem \ref{thm:iid-erasure} is simply the union of the $\cS_t$, $0\leq t\leq \log q$.
We can calculate the apex $a_t$ of $\cS_t$ by recalling that it is the place where all three inequalities are met with equality, resulting in 
\begin{equation}
  a_t = (\tdC_t,\tdQ_t,\tdE_t) 
      = \bigl( (1-\d)(\log q-t), (1-\d)t, \d t \bigr),
\end{equation}
and so $\cS_t = a_t + \cS'$, where 
\begin{equation}
  \cS' := \left\{(\tdC, \tdQ, \tdE)\in\RR^3 \,:\, 
              \tdC + 2\tdQ \leq 0,\,
              \tdQ - \tdE  \leq 0,\,
              \tdC + \tdQ - \tdE  \leq 0 \right\}
\end{equation}
is a simplicial cone rooted at the origin, which crucially is independent of $t$. As the apexes $a_t$ form a straight line connecting 
\begin{align*}
  a_0        &= \bigl( (1-\d)\log q, 0, 0 \bigr) \text{ and} \\
  a_{\log q} &= \bigl( 0, (1-\d)\log q, \d \log q \bigr), 
\end{align*}
we conclude that the region identified in Theorem \ref{thm:iid-erasure} is the union of translates of $\cS_0$ along this line, or equivalently, the convex hull of $\cS_0 \cup \cS_{\log q}$.

To understand it in more detail, let us look at $\cS'$; its apex clearly is the origin, and its extremal rays are determined by saturating with equality any two of the three inequalities. This leads to three infinite half lines,
\begin{align}
  \text{TP} &= \RR_{\geq 0}(-2,1,1), \\
  \text{RC} &= \RR_{\geq 0}(0,-1,-1), \\
  \text{DC} &= \RR_{\geq 0}(2,-1,1). 
\end{align}
Geometrically, we thus have $\cS_t = a_t + \text{TP}+\text{RC}+\text{DC}$ and the entire rate region is simply $[a_0;a_{\log q}] + \text{TP}+\text{RC}+\text{DC}$.

\section{The triple-tradeoff Singleton bound for EACQ codes}
\label{sec:triple-tradeoff-Singleton}
In this section we prove the EACQ Singleton bound, our main contribution in the present paper.
We do so by specializing Theorem~\ref{thm:converse} to the case of a block erasure channel, Equation \eqref{eq:block-erasure-channel-def}, which on a block $A^n$ of $n$ systems $A_i$, uniformly randomly erases $d-1$ out of $n$. The main challenge again is to obtain a single letter characterization of the allowable values of $C$, $Q$ and $E$. We achieve this by using the permutation symmetry of the channel, entropy inequalities as habitual in quantum Shannon theory, and in particular crucially the Lemmas~\ref{lemma:average-entropies} and \ref{lemma:crazy-inequality} introduced in Section \ref{sec:prelim}.
Recall the definition of the quantum block erasure channel,
\begin{equation*}
  \cE_{q,d-1,n}(\rho) = \frac{1}{{n\choose d-1}} \sum_{J\subset[n],\atop |J|=d-1} \left(\tr_{A_J}\rho\right)^{B_{[n]\setminus J}} \ox \proj{\perp}^{B_J},
\end{equation*}
and that an EACQ code has minimum distance $d$ if and only if it achieves error $0$ in transmission over $\cE_{q,d-1,n}$.

\begin{theorem}[EACQ Singleton bound, aka capacity region bound for block erasure channel]
\label{thm:block-erasure}
For the block~erasure channel $\cE_{q,d-1,n}$ and error $\epsilon=0$, the system of converse inequalities from Theorem \ref{thm:converse} for an EACQ code of net $C$ cbits, $Q$ qubits and $E$ ebits, respectively, reduces to the region of triples $(C,Q,E)$ such that there exists a $t \in [0,\log q]$ with
\begin{align}
    C + 2Q &\leq (n-d+1) (\log q + t), \\
    Q - E  &\leq (n-2d+2)t,            \\
    C+Q-E  &\leq (n-d+1) \log q - (d-1) t.
\end{align}
\end{theorem}

\begin{proof}
According to Theorem \ref{thm:converse}, there is a quantum state of the form 
\begin{equation*}
    \s_0^{UA^n{A'}^n} = \sum_{u \in \cU} p(u) \ketbra{u}{u}^U \ox \ph_u^{A^n {A'}^n},
\end{equation*}
which when passed through $\cE_{q,d-1,n}$ becomes
\begin{equation*}
    \s^{UA^nB^n} 
      = \sum_{u \in \cU} p(u) \ketbra{u}{u}^U \ox \cE_{q,d-1,n}^{A'^n \rightarrow B^n} (\ph_u^{A^n A'^n}),
\end{equation*}
and such that
\begin{align*}
  C+2Q  &\leq I(A^nU:B^n)_\sigma = I(U:B^n)_\sigma + I(A^n:B^n|U)_\sigma, \\
  Q-E   &\leq I(A^n\rangle B^nU)_\sigma, \\
  C+Q-E &\leq I(U:B^n)_\sigma + I(A^n\rangle B^nU)_\sigma.
\end{align*}
To evaluate the information quantities occurring on the r.h.s., we write 
\begin{align*}
  \sigma^{UA^n{A'}^n} 
    &= \EE_{J\subset[n],\atop |J|=d-1} \sigma(J)^{UA^n{A'}^n}, \text{ where}\\
  \sigma(J)^{UA^n{A'}^n} 
    &= \sum_{u \in \cU} p(u) \ketbra{u}{u}^U \ox \ph_u^{A^n B_{J^c}} \ox \proj{\perp}^{B_J}. 
\end{align*}
As the $\sigma(J)^{UA^n{A'}^n}$ are orthogonal on $B^n$, we can expand (implicitly fixing $J$ as a uniformly random subset of $[n]$ of size $d-1$):
\begin{align*}
  I(U:B^n)_\sigma 
    &= \EE_J\, I(U:B^n)_{\sigma(J)}, \text{ where} \\
  I(U:B^n)_{\sigma(J)} 
    &= S\left(\EE_U\,\ph_u^{A_{J^c}'}\right) - \EE_U\, S\left(\ph_u^{A_{J^c}'}\right); \\
  I(A^n:B^n|U)_\sigma 
    &= \EE_J \EE_U\, I(A^n:B_{J^c})_{\ph_u}, \text{ where} \\
  I(A^n:B_{J^c})_{\ph_u}
    &= S\left(\ph_u^{A^n}\right) + S\left(\ph_u^{A_{J^c}'}\right) - S\left(\ph_u^{A^nA_{J^c}'}\right) \notag \\
    &= S\left(\ph_u^{{A'}^n}\right) + S\left(\ph_u^{A_{J^c}'}\right) - S\left(\ph_u^{A_J'}\right); \\
  I(A^n\rangle B^nU)_\sigma 
    &= \EE_J \EE_U\, I(A^n\rangle B_{J^c})_{\ph_u}, \text{ where} \\
  I(A^n\rangle B_{J^c})_{\ph_u}
    &= S\left(\ph_u^{A_{J^c}'}\right) - S\left(\ph_u^{A^nA_{J^c}'}\right) \notag \\
    &= S\left(\ph_u^{A_{J^c}'}\right) - S\left(\ph_u^{A_J'}\right). 
\end{align*}

Define two sets of entropy terms, as a function of $\s_0^{UA^n{A'}^n}$, or rather the ensemble $\left\{p(u),\ph_u^{A^n {A'}^n}\right\}$, for integer $1\leq\ell\leq n$. Consider $U$ a random variable distributed according to the law $p(u)$, and let
\begin{align}
  \label{eq:hat-s}
  \widehat{s}_\ell 
      &:= \frac{1}{\ell} \EE_{I\subset[n],\atop |I|=\ell} 
         S\left( \EE_U\, \ph_U^{A_I'} \right) 
        = \frac{1}{\ell}\frac{1}{{n\choose\ell}} \sum_{I\subset[n],\atop |I|=\ell} S(A_I')_{\s_0}, \\
  \label{eq:bar-s}
  \overline{s}_\ell 
      &:= \frac{1}{\ell} \EE_{I\subset[n],\atop |I|=\ell} \EE_U\, S\left(\ph_U^{A_I'}\right)
        = \frac{1}{\ell}\frac{1}{{n\choose\ell}} \sum_{I\subset[n],\atop |I|=\ell} S(A_I'|U)_{\s_0}.
\end{align}
With these definitions we can plug the expanded information quantities into the above, and get 
\begin{align*}
  C+2Q  &\leq I(A^nU:B^n)_\sigma 
         = (n-d+1)\widehat{s}_{n-d+1} + n\overline{s}_n - (d-1)\overline{s}_{d-1}, \\
  Q-E   &\leq I(A^n\rangle B^nU)_\sigma
         = (n-d+1)\overline{s}_{n-d+1} - (d-1)\overline{s}_{d-1}, \\
  C+Q-E &\leq I(U:B^n)_\sigma + I(A^n\rangle B^nU)_\sigma
         = (n-d+1)\widehat{s}_{n-d+1} - (d-1)\overline{s}_{d-1}.
\end{align*}

To conclude the proof, we use certain obvious and known relations between the quantities in Eqs. (\ref{eq:hat-s}) and (\ref{eq:bar-s}). To start,
\begin{align*}
  \forall\ell\quad  0&\leq\overline{s}_\ell\leq\widehat{s}_\ell\leq\log q,\\
  \forall\ell\leq m\quad \overline{s}_m&\leq \overline{s}_\ell,
\end{align*}
the first chain of inequalities by definition of the quantities and concavity of the entropy, the second inequality by Lemma \ref{lemma:average-entropies} (\cite[Lemma~2]{EAQECC-Singleton}). 

Now, for $d-1\leq\frac{n}{2}$, we have by the above $\overline{s}_n \leq \overline{s}_{n-d+1} \leq \overline{s}_{d-1}$, and so we get
\begin{align*}
    C + 2Q &\leq (n-d+1)\log q + (n-d+1)\overline{s}_{d-1}, \\
    Q - E  &\leq (n-2d+2)\overline{s}_{d-1}, \\
    C+Q-E  &\leq (n-d+1) \log q - (d-1)\overline{s}_{d-1}.
\end{align*}
With $t:=\overline{s}_{d-1}\in[0;\log q]$ we arrive at the form of the region claimed in the theorem. 

For $d-1>\frac{n}{2}$, we have to proceed differently, and define $t$ by the relation $(n-d+1)\overline{s}_{n-d+1} - (d-1)\overline{s}_{d-1} =: (n-2d+2)t$. From
Lemma \ref{lemma:crazy-inequality}, we know 
$0 \leq t \leq \overline{s}_{d-1} \leq \log q$; the same lemma tells us
\[
  n\overline{s}_n \leq (n-d+1)\overline{s}_{n-d+1} + (d-1)t.
\]
Thus we get
\begin{align*}
    C + 2Q &\leq (n-d+1)\log q + n\overline{s}_n - (d-1)\overline{s}_{d-1} \\
           &\leq (n-d+1)\log q + (n-d+1)\overline{s}_{n-d+1} + (d-1)t - (d-1)\overline{s}_{d-1} \\
           &=    (n-d+1)\log q + (n-2d+2)t + (d-1)t \\
           &=    (n-d+1)\log q + (n-d+1)t, \\
    Q - E  &\leq (n-2d+2)t, \\
    C+Q-E  &\leq (n-d+1) \log q - (d-1)\overline{s}_{d-1} \\
           &\leq (n-d+1) \log q - (d-1)t,
\end{align*}
concluding the proof.
\end{proof}

\section{Attainability of the bounds: constructions}
\label{sec:attainability}
While the Shannon-theoretic bounds from Theorem \ref{thm:iid-erasure} are met in the i.i.d.~limit by the direct coding theorem from \cite{HsiehWilde}, thus establishing that the region is indeed the capacity region of the erasure channel for any $q$ and and $\delta$, the analogous statement is far from clear for the zero-error coding problem via $\cE_{q,d-1,n}$.

However, from our geometric analysis at the end of Section \ref{sec:triple-tradeoff-erasure} we can gain significant insight into this question as we show in Figures~\ref{fig:attainability}, \ref{fig:attainability:2} and \ref{fig:attainability:3} the Singleton bound region and the attainable portions within it. 

\begin{figure}[ht]
    \input{figures/attanability_lessthan_half.tex}
\caption{Singleton bound region when $\frac{d-1}{n} < \frac{1}{2}$. The blue shaded polytope region indicates the Singleton bound region from Theorem~\ref{thm:block-erasure}. The points $a_0$ and $a_{\log q}$ can be attained using classical MDS codes and EAQMDS codes respectively for large enough $q$. 
Using quantum teleportation (TP), resource conversion (RC) and superdense coding (DC) we can achieve the triangular cones on the boundary of the region. Attainability of the line segment $[a_0, a_{\log q}]$ joining the apexes of the cones (in pink) remains an open question for now. 
Observe that $a_0+\text{TP}$ does not contribute to the extremal edges of the polytope as it is already contained in the convex hull of the other five rays.}
\label{fig:attainability}
\end{figure}
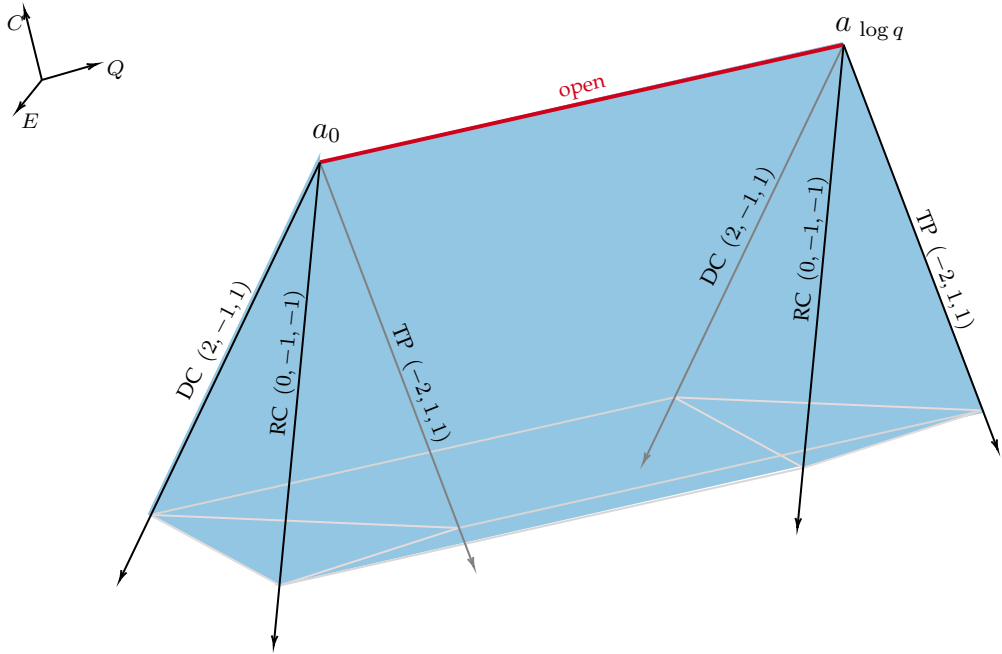

To start with, recall that it has two extreme points, $a_0=\bigl((n-d+1)\log q,0,0\bigr)$ and $a_{\log q}=\bigl(0,(n-d+1)\log q,(d-1)\log q\bigr)$, and we do know that for all sufficiently large $q$, they are both attained by well-known EACQ codes:
\begin{itemize}
    \item In fact, $a_0$ corresponds to a classical MDS code, its $q$-ary symbols encoded into $A$ in an orthonormal basis; it has $Q=E=0$ and transmits $C=(n-d+1)\log q$ cbits, the maximum amount of information for a code capable of correcting $d-1$ erasures. It is well-known that such classical MDS codes exist for all prime powers $q\geq n-1$ (and for some non-prime-powers, too) \cite{MacWilliams-Sloane,Tolhuizen,Ball:MDS-conjecture}.

    \item On the other hand, $a_{\log q}$ corresponds to a ``maximum entanglement EAQECC'' as discussed in \cite[Sec.~V]{EAQECC-Singleton}, where it was linked to the entanglement-assisted quantum capacity of the i.i.d.~erasure channel (EAQ); it has $C=0$, transmits $Q=(n-d+1)\log q$ qubits and and consumes $E=(d-1)\log q$ ebits. In \cite[Sec.~V]{EAQECC-Singleton}, and references therein, it is discussed that for $q \geq n+1$ one can always construct EAQECC codes with arbitrary parameters.
\end{itemize}

\begin{figure}[ht]
    \input{figures/attanability_greaterthan_half.tex}
    \caption{Singleton bound region when $\frac{d-1}{n} > \frac{1}{2}$. Similar to Figure~\ref{fig:attainability}, the convex hull of TP, DC and RC added to the apex points $a_0$ and $a_{\log q}$ constitute the polytope and can be achieved. The segment $[a_0, a_{\log q}]$ (in pink) though is currently not known to be achievable. However, unlike the previous case, it is $a_{\log q}+\text{RC}$ that does not contribute to the extremal edges of the polytope in this case, as it is contained in the convex hull of the other five rays.}
\label{fig:attainability:2}
\end{figure}
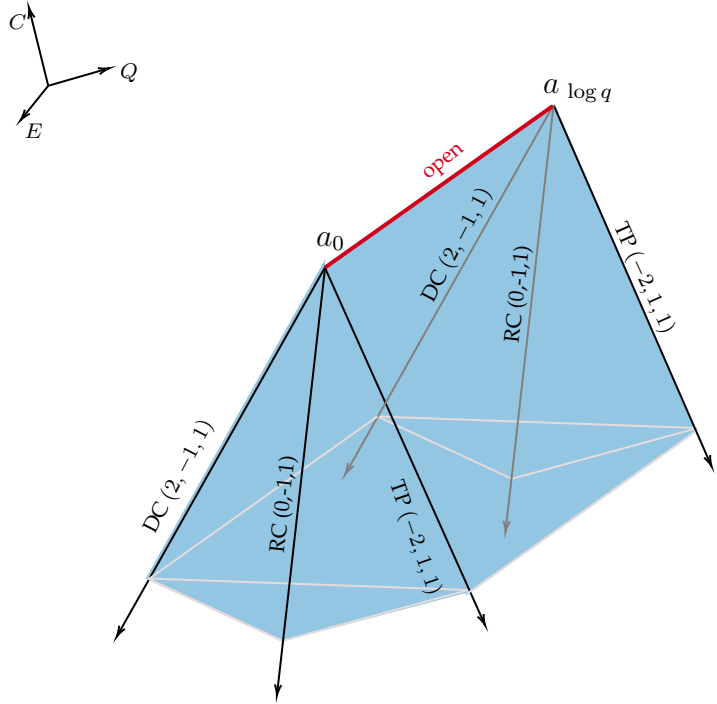

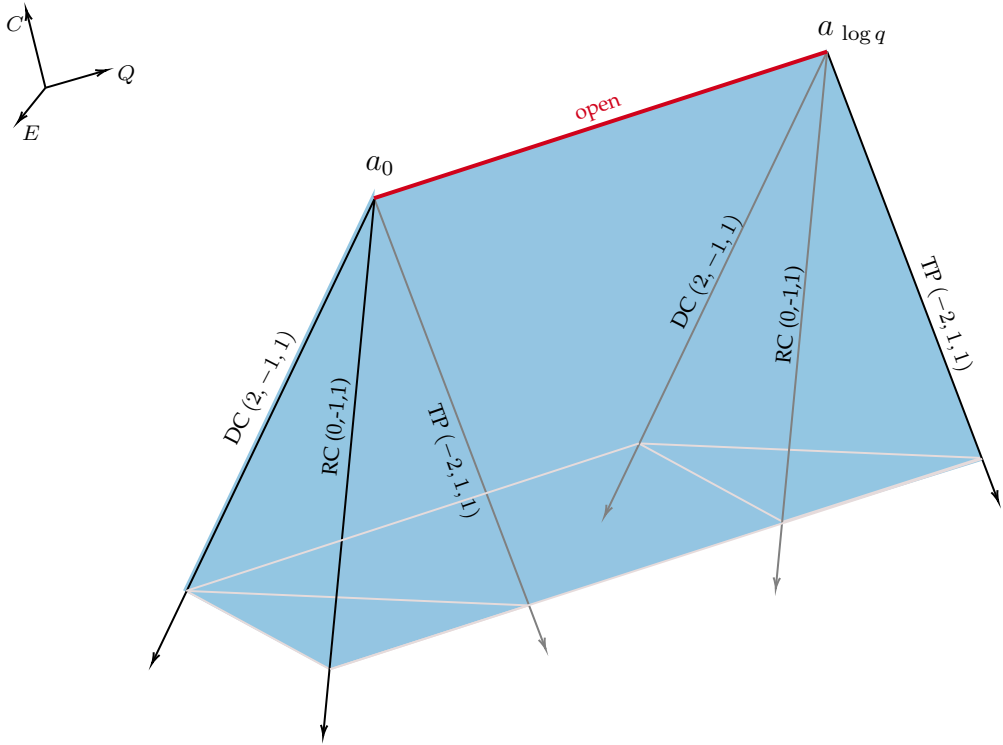
\begin{figure}[ht]
    \input{figures/attanability_equal_to_half.tex}
    \caption{Singleton bound region when $\frac{d-1}{n} = \frac{1}{2}$. Similar to Figures~\ref{fig:attainability} and \ref{fig:attainability:2}, the convex hull of TP, DC and RC added to the apex points $a_0$ and $a_{\log q}$ constitute the polytope and can be achieved. The line segment $[a_0, a_{\log q}]$ (in pink) though is currently not known to be achievable. However, unlike the previous two cases, both $a_0 + \text{TP}$ and $a_{\log q}+\text{RC}$ do not contribute to the extremal edges of the polytope in this case, as they are contained in the convex hull of the other four rays.}
\label{fig:attainability:3}
\end{figure}

Next, for any point $\alpha\in\RR^3$ realised by a given EACQ code, we can actually attain the whole set $\alpha + \cS' = \alpha + \text{TP} + \text{RC} + \text{DC}$ (or at least the points in this set corresponding to log-integer coordinates). Namely, note that the one-shot rate triple $(-2,1,1)$ comes from the protocol of quantum teleportation ($2$ cbits and $1$ ebit are consumed to transmit $1$ qubit); the triple $(0,-1,-1)$ represents the resource conversion of qubits into ebits at unit rate; finally, the triple $(2,-1,1)$ comes from the protocol of dense coding ($1$ qubit and $1$ ebit are consumed to transmit $2$ cbits). Thus, by concatenating the given EACQ code with suitable amounts of teleportation, dense coding and resource conversion, we can attain every reasonable rate triple of the form $\alpha + x(-2,1,1) + y(0,-1,-1) + z(2,-1,1)$ with $x,y,z\geq 0$ for large enough $q$. 
Depending on whether $\frac{d-1}{n}$ is greater than or less than $\frac{1}{2}$, we observe that only a particular subset of five from the total six possible combinations of the three protocols TP, RC and RC acting on $a_0$ and $a_{\log q}$ result in extremal rays of the polytope. 
When $\frac{d-1}{n} = \frac{1}{2}$, only four of the six combinations result in extremal rays of the polytope.
Note that throughout we needed the alphabet size $q$ large enough so that we are able to construct EACQ codes attaining $a_0$ and $a_{\log q}$. 

This means that the attainability of the region from Theorem \ref{thm:block-erasure} is reduced to that of the line segment $[a_0,a_{\log q}]$. We have to leave this as an open question, but it is curious that if the line were attained, the realising codes would have to be some kind of interpolation between purely classical MDS codes and fully quantum EAQ codes. As a somewhat separate, but at the same time preparatory question to this, we would like to know what we can deduce about codes attaining the boundary of our region, and in particular necessary conditions for attaining the line segment $[a_0,a_{\log q}]$; cf.~\cite{EAQECC-Singleton}.

Returning to the i.i.d.~channel case $\cE_{q,\d}^{\ox n}$ with asymptotically large block length and asymptotically small error, we note that this is indeed possible. Not only are the points $a_0$ and $a_{\log q}$ attainable by capacity-achieving classical codes and entanglement-assisted quantum capacity, respectively, and in fact for arbitrary alphabet size $q$ (not just large enough), the line segment linking these two points is attained by the time sharing principle, where we subdivide the block of $n$ channel uses into $\lambda n$ and $(1-\lambda)n$ uses on which the MDS and the EAQ code are realised, respectively. This works because in the Shannon-theoretic setting we are happy to make the small error of not correcting non-typical erasure patterns, rather we focus on $\sim \delta\lambda n$ ($\sim \delta(1-\lambda)n$) erasures in the first (second) block, respectively. In the coding-theoretic (zero-error) setting we do not have that luxury, and the code needs to prepare for any distribution of $\delta n$ erasure errors.

\section{Conclusion}
\label{sec:conclusion}
We have shown that one can adapt the information theoretic converse proofs of Hsieh and Wilde \cite{HsiehWilde} and of Wilde \cite{wildebeast} for the triple-tradeoff region of communication over a general channel to the one-shot setting. Applying the obtained converse bounds to the one-shot zero-error case of a block erasure channel, we have derived the Singleton bounds for EACQ codes. By specialising to the hyperplane $C=0$, we recover the region found in \cite{EAQECC-Singleton}, and by specialising further to $C=E=0$, we recover the original quantum Singleton bound for QECC, $Q \leq \max\{0,n-2d+2\}$. 

In \cite{EAQECC-Singleton}, the question of attainability of the whole region for sufficiently large alphabet $q$ had been left open, which boiled down to the line connecting the points EAQ and ``MDS'' in \cite[Fig.~3(c)]{EAQECC-Singleton}; EAQ there means the same as here, as the protocol for entanglement-assisted quantum  capacity has $C=0$, but the ``MDS'' there is really our MDS here concatenated with teleportation. Thus, the question of attainability of the line EAQ-``MDS'' in \cite{EAQECC-Singleton} is lifted to that of the line segment connecting $a_{\log q}$ (EAQ) with $a_0$ (MDS) in the three-dimensional rate region, which we think might be a much clearer question, as it is about interpolating between an essentially classical code and a fully quantum code.

\acknowledgments
The authors thank S. Dedalus for ambiguous advice that led to the proof of a curious theorem, which however ultimately we could not include in the present paper. Furthermore, we thank the anonymous referees of the conference version of this work \cite{EACQ-Singleton-ISIT} of and the present journal version, whose comments helped improve the presentation of the present paper. 
AW is supported by the European Commission QuantERA grant ExTRaQT (Spanish MICINN project PCI2022-132965), by the Spanish MINECO (project PID2019-107609GB-I00) with the support of FEDER funds, by the Spanish MICINN with funding from European Union NextGenerationEU (PRTR-C17.I1) and the Generalitat de Catalunya, and by the Alexander von Humboldt Foundation, as well as the Institute of Advanced Study of the Technical University Munich.

\bibliographystyle{unsrt}



\appendix

\section{Proof of Theorem~\ref{thm:converse}}
\label{app:converse}

This is essentially the converse proof of Wilde in \cite[Ch.~25.4]{wildebeast}, only that we consider a general (non-i.i.d.) channel and use one-shot rates. To get our error-dependent additive constants, we use Lemmas \ref{lemma:Fannes} and \ref{lemma:AF}.
For the sake of self-containedness, we reproduce the argument here in full. Crucially, we use information theoretic deductions to express the relation among $C$, $Q$, $E$ in terms of the channel output $\s$. For that, we identify the channel input $(u, \ph_u^{AA'})$ as the encoded state $(mv, \o_{mv}^{RT_A'AWT_B})$ [cf.~Equation~\eqref{def:O}] of the problem setup. The $\s^{UAB}$ obtained on passing $\ph_u^{AA'}$ through noisy channel $\mathcal{N}$ would correspond to $\barO^{MVRT_A'BWT_B}$ [cf.~Equation~\eqref{def:barO}]. We identify the classical component $U$ (of $\ph_u^{AA'}$) $\equiv MV$ (of $\barO^{MVRT_A'BWT_B}$), the uncorrupted quantum component $A$ ( of $\ph_u^{AA'}$) $\equiv RT_A'WT_B$ (of $\barO^{MVRT_A'BWT_B}$) and the corrupted quantum component $A'$ (of $\ph_u^{AA'}$) $\equiv A$ (of $\barO^{MVRT_A'BWT_B}$). We start by examining the information contained in  $\overline{\Phi}$ and $\Phi$.
\begin{align}
    C_2 + 2Q_2 &\stackrel{}{=} I(M:\Mhat)_{\overline{\Phi}} + I(R:\Xhat)_{\Phi} \label{eq:converse:pA:1}\\
    &= I(MR:\Mhat\Xhat)_{\barG_{\text{ideal}}} \label{eq:converse:pA:2}\\
    &\stackrel{}{\leq} I(MR:\Mhat\Xhat)_{\barG} +  \e' \label{eq:converse:pA:3}\\
    &\stackrel{}{\leq} I(MR:BWT_BV)_{\barO} +  \e' \label{eq:converse:pA:4}\\
    &\stackrel{}{=} I(MR:T_B)_{\barO} + I(MR:BVW|T_B)_{\barO} + \e' \label{eq:converse:pA:5}\\
    &\stackrel{}{\leq} 0 + I(MRT_B:BWV)_{\barO} + \e' \label{eq:converse:pA:6}\\
    &\stackrel{}{\leq} I(MRT_A'T_B:BWV)_{\barO} + \e' \label{eq:converse:pA:7}\\
    &\stackrel{}{=} I(MRT_A'T_B:B)_{\barO} + I(MRT_A'T_B:V|B)_{\barO} +I(MRT_A'T_B:W|VB)_{\barO} +  \e' \label{eq:converse:pA:8}\\
    &\stackrel{}{\leq} I(MVRT_A'T_BW:B)_{\barO} + I(MRT_A'T_BB:V)_{\barO} +I(MRT_A'T_BB:W|V)_{\barO} + \e' \label{eq:converse:pA:9}\\
    &\stackrel{}{\leq} I(UA:B)_{\s} + S(V)_{\s} + 2S(W)_{\s} + \e' \label{eq:converse:pA:10}\\
    &\stackrel{}{=} I(UA:B)_{\s} + C_1 + 2Q_1 + \e'. \label{eq:converse:pA:11}
\end{align}
Here, Equation~\eqref{eq:converse:pA:1} follows by evaluating quantum mutual information of the perfectly correlated classical state $\overline{\Phi}^{M\widehat{M}}$ and the maximally entangled quantum state $\Phi^{R\Xhat}$. In Equation~\eqref{eq:converse:pA:2}, we reduce the right hand side using the fact that $\barG_{\text{ideal}}^{M \Mhat R \Xhat}$ from its definition [cf.~Equation~\eqref{eq:def:barGideal}] is a product state of $\overline{\Phi}^{M\widehat{M}}$ and $\Phi^{R\Xhat}$. The given error $\e$ of this EACQ code, from its definition in Equation~\eqref{eq:distancebound}, corresponds to the upper bound on distance between $\barG^{M \Mhat R \Xhat}$ and $\barG_{\text{ideal}}^{M \Mhat R \Xhat}$. Invoking Lemma \ref{lemma:AF}, we get Inequality~\eqref{eq:converse:pA:3} with $\e' := 2\e (C_2 + Q_2) + g(\e)$.
Inequality \eqref{eq:converse:pA:4} follows from quantum data processing. We get Equation~\eqref{eq:converse:pA:5} using the chain rule for mutual information. In Inequality~\eqref{eq:converse:pA:6}, the first term evaluates to zero due to the independence of the starting states $M$, $R$ and $T_B$ in the problem setup; the second term is relaxed by adding a non-negative term $I(T_B:BWV)_{\bar)}$. Inequality~\eqref{eq:converse:pA:7} is from quantum data processing. Equation~\eqref{eq:converse:pA:8} again comes from the chain rule for mutual information. In Inequality~\eqref{eq:converse:pA:9}, we use the quantum data processing inequality in the first term and add non-negative quantities $I(B:V)_{\barO}$, $I(B:W|V)_{\barO}$ to the second and third terms. Inequality~\eqref{eq:converse:pA:10} follows from identifying $U\equiv MV$, $A \text{ (of $\s$)} \equiv RT_A'WT_B$, $B \text{ (of $\s$)} \equiv B$ and applying information theoretic deductions. We then evaluate the information content of the noiselessly transmitted classical message $V$ and the quantum message $W$.

Next,
\begin{align}
    Q_2 + E_2 &\stackrel{}{=} I(R\rangle\Xhat)_{\Phi} + I(T_A'\rangle T_B')_{\ph} \label{eq:converse:pB:1} \\
    &\stackrel{}{=} I(RT_A'\rangle\Xhat T_B' M \Mhat)_{\barG_{\text{ideal}}} \label{eq:converse:pB:2}\\
    &\stackrel{}{\leq} I(RT_A'\rangle\Xhat T_B' M \Mhat)_{\barG} + \e'' \label{eq:converse:pB:3} \\
    &\stackrel{}{\leq} I(RT_A'\rangle\Xhat T_B' M V \Mhat)_{\barG} + \e'' \label{eq:converse:pB:3a} \\
    &\stackrel{}{\leq} I(RT_A'\rangle B W T_B M V)_{\barO} + \e'' \label{eq:converse:pB:4}\\
    &\stackrel{}{=} I(RT_A' W T_B \rangle B M V)_{\barO} + S(WT_B | B MV)_{\barO} + \e'' \label{eq:converse:pB:5}\\
    &\stackrel{}{=} I(RT_A' W T_B\rangle B M V)_{\barO} + S(T_B | B MV)_{\barO} + S(W|T_B BMV)_{\barO} + \e'' \label{eq:converse:pB:6}\\
    &\stackrel{}{\leq} I(RT_A' W T_B\rangle B M V)_{\barO} + S(T_B)_{\barO} + S(W)_{\barO} + \e''\label{eq:converse:pB:7}\\
    &\stackrel{}{=} I(A\rangle B U)_{\s} + E_1 + Q_1 + \e''. \label{eq:converse:pB:8} 
\end{align}
In Equation~\eqref{eq:converse:pB:1}, we express the coherent information of the pure states $\Ph$ and $\ph$. Equality~\eqref{eq:converse:pB:2} comes from the definition of $\barG_{\text{ideal}}$ [cf.~Equation~\eqref{eq:def:barGideal}], which is a product of the earlier states. In Inequality~\eqref{eq:converse:pB:3} we use a similar technique as in the previous reduction. We employ the error $\e$ of the EACQ code, which is the upper bound of the distance between $\barG^{M \Mhat R \Xhat T_A' T_B'}$ and $\barG_{\text{ideal}}^{M \Mhat R \Xhat T_A' T_B'}$ and then invoke Lemma \ref{lemma:AF}, with $\e'' := 2\e (Q_2 + |T_A'|) + g(\e)$. Inequality\eqref{eq:converse:pB:3a} follows from strong subadditivity. Then we use quantum data processing in Inequality~\eqref{eq:converse:pB:4} and then expand the resulting terms in Eqs.~\eqref{eq:converse:pB:5} and \eqref{eq:converse:pB:6}. We add non-negative terms $I(T_B:BMV)_{\barO}$, $I(W:T_BBMV)_{\barO}$ in Equation~\eqref{eq:converse:pB:7}. We identify the states of $\barO$ as those in $\s$, $U \equiv MV$, $A \text{ (of $\s$)} \equiv RT_A'WT_B$, $B \text{ (of $\s$)} \equiv B$. In the final step, Equation~\eqref{eq:converse:pB:8}, we substitute the amount of information that is contained in $T_B$ and $W$.    

Finally,
\begin{align}
    C_2 + Q_2 + E_2 &\stackrel{}{=} I(M:\Mhat)_{\barG_{\text{ideal}}} + I(RT_A'\rangle \Xhat T_B')_{\barG_{\text{ideal}}} \label{eq:converse:pC:1}\\
    &\stackrel{}{=} I(M:\Xhat T_B' \Mhat)_{\barG_{\text{ideal}}} + I(RT_A'\rangle\Xhat T_B' M \Mhat)_{\barG_{\text{ideal}}} \label{eq:converse:pC:2}\\
    &\stackrel{}{\leq} I(M:\Xhat T_B' \Mhat)_{\barG} + I(RT_A'\rangle\Xhat T_B' M \Mhat)_{\barG} + \e'''
    \label{eq:converse:pC:3}\\
    &\stackrel{}{\leq} I(MV:\Xhat T_B' \Mhat)_{\barG} + I(RT_A'\rangle\Xhat T_B' M V \Mhat)_{\barG} + \e'''
    \label{eq:converse:pC:3a}\\
    &\stackrel{}{\leq} I(MV:B W T_B V)_{\barO} + I(RT_A'\rangle B W T_B M V )_{\barO} + \e'''
    \label{eq:converse:pC:4}\\
    &\stackrel{}{=} I(MV:B W T_B)_{\barO} + I(MV:V|B W T_B)_{\barO} \notag \\
    & \quad\quad + I(RT_A'WT_B\rangle B M V )_{\barO} + S(WT_B | B MV)_{\barO} + \e'''
    \label{eq:converse:pC:5}\\
    &\stackrel{}{\leq} I(MV:B)_{\barO} + I(MV: W T_B|B)_{\barO} + I(MVBWT_B:V)_{\barO} \notag\\
    &\quad\quad + I(RT_A'WT_B\rangle B M V )_{\barO} + S(WT_B | B MV)_{\barO} + \e'''
    \label{eq:converse:pC:6}\\
    &\stackrel{}{=} I(MV:B)_{\barO} + I(RT_A'WT_B\rangle BMV)_{\barO} + S(MV|B)_{\barO} + S(WT_B|B)_{\barO}\notag \\ 
    &\quad\quad - S(WT_B MV|B)_{\barO} + C_1  + S(WT_B | B MV)_{\barO} + \e'''
    \label{eq:converse:pC:7}\\
    &\stackrel{}{\leq} I(MV:B)_{\barO} + I(RT_A'WT_B\rangle BMV)_{\barO} + S(W)_{\barO} + S(T_B)_{\barO} + C_1 + \e''' \label{eq:converse:pC:8} \\
    &\stackrel{}{=} I(U:B)_{\s} + I(A \rangle B U)_{\s} + Q_1 + E_1 + C_1 + \e'''. \label{eq:converse:pC:9}
\end{align}
Here too, we start by expressing the information content of the messages. Equality~\eqref{eq:converse:pC:1} follows from the information content of the correlated classical state $\overline{\Phi}^{M\Mhat}$ and Equation~\eqref{eq:converse:pB:2} of the previous reduction. In Equation~\eqref{eq:converse:pC:2}, we expand the terms further since $\barG_{\text{ideal}}$ is a product state. Like previous reductions, we employ Lemma \ref{lemma:AF} and use the upper bound $\e$ on the trace distance between $\barG^{M \Mhat R \Xhat T_A' T_B'}$ and $\barG_{\text{ideal}}^{M \Mhat R \Xhat T_A' T_B'}$ to get Inequality \eqref{eq:converse:pC:3} with $\e''' := 2\e (C_2 + Q_2 + |T_A'|) + 2g(\e) $.
Inequality \eqref{eq:converse:pC:3a} comes from the data processing inequality of quantum mutual information. Inequality~\eqref{eq:converse:pC:4} follows from the quantum data processing property. In Equation~\eqref{eq:converse:pC:5} we expand both the terms using chain rule. Inequality~\eqref{eq:converse:pC:6} comes from expanding the first term again using information theoretic reductions and adding to the second term a non-negative quantity $I(V:BWT_B)_{\barO}$. In the next step, Equation~\eqref{eq:converse:pC:7}, we rearrange the terms, expand the second term, and evaluate an upper bound on the third term $I(MVBWT_B:V)_{\barO} \leq S(V)_{\barO}\leq C_1$. The third, fifth and seventh terms get cancelled and the fourth term gets relaxed in Inequality~\eqref{eq:converse:pC:8}. Finally, we identify the states of $\barO$ as those in $\s$, $U \equiv MV$, $A \text{ (of $\s$)} \equiv RT_A'WT_B$, $B \text{ (of $\s$)} \equiv B$ and evaluate the information quantities of the remaining terms. This completes the proof.
\qed

\section{Proof of Theorem \ref{thm:iid-erasure-eliminated}}
\label{app:iid-erasure-eliminated}

The proof rests on Fourier-Motzkin elimination of the parameter $t$ in the set of constraints from Theorem \ref{thm:iid-erasure}. We need to distinguish between the cases $\delta<\frac12$, $\delta=\frac12$ and $\delta>\frac12$ as these determine the signs of certain coefficients in the inequalities.

We start by using the Inequalities \eqref{eq:erasure:C+2Q} and \eqref{eq:erasure:C+Q-E} to get
\begin{align}
        \frac{\tdC + 2\tdQ}{1-\d} - \log q &\leq t, \label{eq:iid-erasure-eliminated:1}\\
        \frac{(1-\d) \log q}{\d} - \frac{\tdC + \tdQ - \tdE}{\d} & \geq t. \label{eq:iid-erasure-eliminated:2}
        \intertext{From the range of $t$ we also have}
        t &\geq 0, \label{eq:iid-erasure-eliminated:3}\\
        t &\leq \log q. \label{eq:iid-erasure-eliminated:4}
\end{align}
Combining Inequalities \eqref{eq:iid-erasure-eliminated:1} and \eqref{eq:iid-erasure-eliminated:4}, we get \eqref{eq:iid-erasure-eliminated:thm:1}. 
Combining \eqref{eq:iid-erasure-eliminated:2} and \eqref{eq:iid-erasure-eliminated:3} gives \eqref{eq:iid-erasure-eliminated:thm:3}; and \eqref{eq:iid-erasure-eliminated:1} and \eqref{eq:iid-erasure-eliminated:2} imply \eqref{eq:iid-erasure-eliminated:thm:4}.
For the remaining inequalities we analyze over two cases when $\d < \frac{1}{2}$ or $\d \geq \frac{1}{2}$. In the first case, Inequality~\eqref{eq:erasure:Q-E} reduces to
\begin{align}
        \frac{\tdQ - \tdE}{1-2\d} &\leq t. \label{eq:iid-erasure-eliminated:5}
\end{align}
Combining Inequalities \eqref{eq:iid-erasure-eliminated:2} and \eqref{eq:iid-erasure-eliminated:5} proves the second part of Inequality \eqref{eq:iid-erasure-eliminated:thm:2}. Combining \eqref{eq:iid-erasure-eliminated:4} and \eqref{eq:iid-erasure-eliminated:5} proves the first part of Inequality \eqref{eq:iid-erasure-eliminated:thm:5}. In the case when $\d > \frac{1}{2}$, Inequality \eqref{eq:erasure:Q-E} reduces to
\begin{align}
        \frac{\tdQ - \tdE}{2\d - 1} &\geq t. \label{eq:iid-erasure-eliminated:6}
\end{align}
Combining Inequalities \eqref{eq:iid-erasure-eliminated:1} and \eqref{eq:iid-erasure-eliminated:6} proves the first part of Inequality \eqref{eq:iid-erasure-eliminated:thm:2}. Combining \eqref{eq:iid-erasure-eliminated:3} and \eqref{eq:iid-erasure-eliminated:6} proves the second part of Inequality \eqref{eq:iid-erasure-eliminated:thm:5}. For the case when $\d = \frac{1}{2}$, Inequality \eqref{eq:erasure:Q-E} becomes simply
\begin{align}
        \tdQ - \tdE & \leq 0.
\end{align}
This is exactly what both the parts of Inequality \eqref{eq:iid-erasure-eliminated:thm:5} and the first part of Inequality \eqref{eq:iid-erasure-eliminated:thm:2} reduce to, as well. This proves all the inequalities in Theorem~\ref{thm:iid-erasure-eliminated}. To see why these inequalities fully characterize the capacity region, one can check that all other inequality combinations from \eqref{eq:iid-erasure-eliminated:1} to \eqref{eq:iid-erasure-eliminated:6} only result in bounds that are implied by those we have included above. We omit this straightforward confirmation. 
\qed

\end{document}

%% file: figures/attanability_lessthan_half.tex
\tikzset{every picture/.style={line width=0.75pt}} 

\begin{tikzpicture}[x=0.75pt,y=0.75pt,yscale=-0.45,xscale=0.45]
\scriptsize
\draw  [color={rgb, 255:red, 147; green, 198; blue, 226 }  ,draw opacity=1 ][fill={rgb, 255:red, 147; green, 198; blue, 226 }  ,fill opacity=1 ] (481.16,239.53) -- (436.09,716.19) -- (292.51,637.77) -- cycle ;
\draw  [color={rgb, 255:red, 147; green, 198; blue, 226 }  ,draw opacity=1 ][fill={rgb, 255:red, 147; green, 198; blue, 226 }  ,fill opacity=1 ] (1062.18,114.61) -- (1216.4,522.29) -- (1017.06,584.8) -- cycle ;
\draw  [color={rgb, 255:red, 147; green, 198; blue, 226 }  ,draw opacity=1 ][fill={rgb, 255:red, 147; green, 198; blue, 226 }  ,fill opacity=1 ] (479.8,245.14) -- (1059.57,112.14) -- (873.34,503.72) -- (293.58,636.72) -- cycle ;
\draw  [color={rgb, 255:red, 147; green, 198; blue, 226 }  ,draw opacity=1 ][fill={rgb, 255:red, 147; green, 198; blue, 226 }  ,fill opacity=1 ] (1059.56,112.14) -- (481.3,244.79) -- (639.8,656.19) -- (1218.06,523.54) -- cycle ;
\draw  [color={rgb, 255:red, 147; green, 198; blue, 226 }  ,draw opacity=1 ][fill={rgb, 255:red, 147; green, 198; blue, 226 }  ,fill opacity=1 ] (478.7,245.39) -- (1059.55,112.15) -- (1015.12,583.22) -- (434.27,716.47) -- cycle ;
\draw [color={rgb, 255:red, 230; green, 219; blue, 219 }  ,draw opacity=1 ]   (294.28,638.81) -- (437.8,717.65) ;
\draw [color={rgb, 255:red, 230; green, 219; blue, 219 }  ,draw opacity=1 ]   (437.8,717.65) -- (636,653.41) ;
\draw [color={rgb, 255:red, 230; green, 219; blue, 219 }  ,draw opacity=1 ]   (294.28,638.81) -- (636,653.41) ;
\draw [color={rgb, 255:red, 230; green, 219; blue, 219 }  ,draw opacity=1 ]   (875.2,507.43) -- (1018.72,586.26) ;
\draw [color={rgb, 255:red, 230; green, 219; blue, 219 }  ,draw opacity=1 ]   (1018.72,586.26) -- (1216.91,522.03) ;
\draw [color={rgb, 255:red, 230; green, 219; blue, 219 }  ,draw opacity=1 ]   (875.2,507.43) -- (1216.91,522.03) ;
\draw [color={rgb, 255:red, 208; green, 2; blue, 27 }  ,draw opacity=1 ][line width=1.5]    (482.23,244.66) -- (1063.14,113.28) ;
\draw [color={rgb, 255:red, 216; green, 213; blue, 213 }  ,draw opacity=1 ]   (294.28,638.81) -- (875.2,507.43) ;
\draw [color={rgb, 255:red, 216; green, 213; blue, 213 }  ,draw opacity=1 ]   (636,653.41) -- (1216.91,522.03) ;
\draw [color={rgb, 255:red, 216; green, 213; blue, 213 }  ,draw opacity=1 ]   (437.8,717.65) -- (1018.72,586.26) ;
\draw    (482.23,244.66) -- (260.97,710) ;
\draw [shift={(260.11,711.81)}, rotate = 295.43] [color={rgb, 255:red, 0; green, 0; blue, 0 }  ][line width=0.75]    (10.93,-3.29) .. controls (6.95,-1.4) and (3.31,-0.3) .. (0,0) .. controls (3.31,0.3) and (6.95,1.4) .. (10.93,3.29)   ;
\draw [color={rgb, 255:red, 128; green, 128; blue, 128 }  ,draw opacity=1 ]   (482.23,244.66) -- (652.38,695.34) ;
\draw [shift={(653.08,697.21)}, rotate = 249.32] [color={rgb, 255:red, 128; green, 128; blue, 128 }  ,draw opacity=1 ][line width=0.75]    (10.93,-3.29) .. controls (6.95,-1.4) and (3.31,-0.3) .. (0,0) .. controls (3.31,0.3) and (6.95,1.4) .. (10.93,3.29)   ;
\draw    (482.23,244.66) -- (431.16,782.81) ;
\draw [shift={(430.97,784.8)}, rotate = 275.42] [color={rgb, 255:red, 0; green, 0; blue, 0 }  ][line width=0.75]    (10.93,-3.29) .. controls (6.95,-1.4) and (3.31,-0.3) .. (0,0) .. controls (3.31,0.3) and (6.95,1.4) .. (10.93,3.29)   ;
\draw [color={rgb, 255:red, 128; green, 128; blue, 128 }  ,draw opacity=1 ]   (1063.14,113.28) -- (841.89,578.62) ;
\draw [shift={(841.03,580.42)}, rotate = 295.43] [color={rgb, 255:red, 128; green, 128; blue, 128 }  ,draw opacity=1 ][line width=0.75]    (10.93,-3.29) .. controls (6.95,-1.4) and (3.31,-0.3) .. (0,0) .. controls (3.31,0.3) and (6.95,1.4) .. (10.93,3.29)   ;
\draw    (1063.14,113.28) -- (1233.29,563.95) ;
\draw [shift={(1234,565.82)}, rotate = 249.32] [color={rgb, 255:red, 0; green, 0; blue, 0 }  ][line width=0.75]    (10.93,-3.29) .. controls (6.95,-1.4) and (3.31,-0.3) .. (0,0) .. controls (3.31,0.3) and (6.95,1.4) .. (10.93,3.29)   ;
\draw    (1063.14,113.28) -- (1012.07,651.42) ;
\draw [shift={(1011.88,653.41)}, rotate = 275.42] [color={rgb, 255:red, 0; green, 0; blue, 0 }  ][line width=0.75]    (10.93,-3.29) .. controls (6.95,-1.4) and (3.31,-0.3) .. (0,0) .. controls (3.31,0.3) and (6.95,1.4) .. (10.93,3.29)   ;
\draw    (173.82,152.69) -- (154.98,75.92) ;
\draw [shift={(154.5,73.98)}, rotate = 76.21] [color={rgb, 255:red, 0; green, 0; blue, 0 }  ][line width=0.75]    (10.93,-3.29) .. controls (6.95,-1.4) and (3.31,-0.3) .. (0,0) .. controls (3.31,0.3) and (6.95,1.4) .. (10.93,3.29)   ;
\draw    (173.82,152.69) -- (232.91,135.49) ;
\draw [shift={(234.83,134.93)}, rotate = 163.77] [color={rgb, 255:red, 0; green, 0; blue, 0 }  ][line width=0.75]    (10.93,-3.29) .. controls (6.95,-1.4) and (3.31,-0.3) .. (0,0) .. controls (3.31,0.3) and (6.95,1.4) .. (10.93,3.29)   ;
\draw    (173.82,152.69) -- (147.4,185.73) ;
\draw [shift={(146.15,187.3)}, rotate = 308.64] [color={rgb, 255:red, 0; green, 0; blue, 0 }  ][line width=0.75]    (10.93,-3.29) .. controls (6.95,-1.4) and (3.31,-0.3) .. (0,0) .. controls (3.31,0.3) and (6.95,1.4) .. (10.93,3.29)   ;

\draw (470.28,198.55) node [anchor=north west][inner sep=0.75pt]  [font=\large] [align=left] {$\displaystyle a_{0}$};
\draw (1050.44,79.69) node [anchor=north west][inner sep=0.75pt]  [font=\large] [align=left] {$\displaystyle a\ _{\log q}$};
\draw (131.6,77.78) node [anchor=north west][inner sep=0.75pt]   [align=left] {$\displaystyle C$};
\draw (242.54,127.76) node [anchor=north west][inner sep=0.75pt]   [align=left] {$\displaystyle Q$};
\draw (147.46,187.19) node [anchor=north west][inner sep=0.75pt]   [align=left] {$\displaystyle E$};
\draw (316.67,504.27) node [anchor=north west][inner sep=0.75pt]  [rotate=-297.34] [align=left] { DC\; $\displaystyle ( 2,-1,1)$};
\draw (421.56,551.59) node [anchor=north west][inner sep=0.75pt]  [rotate=-276.16] [align=left] { RC\; $\displaystyle ( 0,-1,-1)$};
\draw (579.9,418.6) node [anchor=north west][inner sep=0.75pt]  [rotate=-68.24] [align=left] { TP\; $\displaystyle ( -2,1,1)$};
\draw (897.49,377.3) node [anchor=north west][inner sep=0.75pt]  [rotate=-297.34] [align=left] {DC\; $\displaystyle ( 2,-1,1)$};
\draw (1002.39,424.62) node [anchor=north west][inner sep=0.75pt]  [rotate=-276.16] [align=left] {RC\; $\displaystyle ( 0,-1,-1)$};
\draw (1160.72,291.63) node [anchor=north west][inner sep=0.75pt]  [rotate=-68.24] [align=left] {TP\; $\displaystyle ( -2,1,1)$};
\draw (742.62,157.81) node [anchor=north west][inner sep=0.75pt]  [rotate=-348.04] [align=left] {\textcolor[rgb]{0.82,0.01,0.11}{open}};
\end{tikzpicture}

%% file: figures/attanability_greaterthan_half.tex
\tikzset{every picture/.style={line width=0.75pt}} 

\begin{tikzpicture}[x=0.75pt,y=0.75pt,yscale=-0.5,xscale=0.5]
\scriptsize
\draw    (76.82,100.69) -- (57.98,23.92) ;
\draw [shift={(57.5,21.98)}, rotate = 76.21] [color={rgb, 255:red, 0; green, 0; blue, 0 }  ][line width=0.75]    (10.93,-3.29) .. controls (6.95,-1.4) and (3.31,-0.3) .. (0,0) .. controls (3.31,0.3) and (6.95,1.4) .. (10.93,3.29)   ;
\draw    (76.82,100.69) -- (135.91,83.49) ;
\draw [shift={(137.83,82.93)}, rotate = 163.77] [color={rgb, 255:red, 0; green, 0; blue, 0 }  ][line width=0.75]    (10.93,-3.29) .. controls (6.95,-1.4) and (3.31,-0.3) .. (0,0) .. controls (3.31,0.3) and (6.95,1.4) .. (10.93,3.29)   ;
\draw    (76.82,100.69) -- (50.4,133.73) ;
\draw [shift={(49.15,135.3)}, rotate = 308.64] [color={rgb, 255:red, 0; green, 0; blue, 0 }  ][line width=0.75]    (10.93,-3.29) .. controls (6.95,-1.4) and (3.31,-0.3) .. (0,0) .. controls (3.31,0.3) and (6.95,1.4) .. (10.93,3.29)   ;
\draw  [color={rgb, 255:red, 147; green, 198; blue, 226 }  ,draw opacity=1 ][fill={rgb, 255:red, 147; green, 198; blue, 226 }  ,fill opacity=1 ] (354.41,285.31) -- (582.57,121.9) -- (722.71,448.54) -- (494.55,611.95) -- cycle ;
\draw  [color={rgb, 255:red, 147; green, 198; blue, 226 }  ,draw opacity=1 ][fill={rgb, 255:red, 147; green, 198; blue, 226 }  ,fill opacity=1 ] (353.06,284.75) -- (497.39,608.83) -- (310.76,658.52) -- cycle ;
\draw  [color={rgb, 255:red, 147; green, 198; blue, 226 }  ,draw opacity=1 ][fill={rgb, 255:red, 147; green, 198; blue, 226 }  ,fill opacity=1 ] (352.03,279.56) -- (309.84,658.46) -- (175.42,596.09) -- cycle ;
\draw [color={rgb, 255:red, 230; green, 219; blue, 219 }  ,draw opacity=1 ]   (177.1,596.95) -- (311.43,659.62) ;
\draw [color={rgb, 255:red, 230; green, 219; blue, 219 }  ,draw opacity=1 ]   (405.38,433.35) -- (539.72,496.02) ;
\draw [color={rgb, 255:red, 230; green, 219; blue, 219 }  ,draw opacity=1 ]   (539.72,496.02) -- (725.24,444.96) ;
\draw [color={rgb, 255:red, 230; green, 219; blue, 219 }  ,draw opacity=1 ]   (405.38,433.35) -- (725.24,444.96) ;
\draw    (353.02,283.62) -- (146.09,653.23) ;
\draw [shift={(145.11,654.98)}, rotate = 299.24] [color={rgb, 255:red, 0; green, 0; blue, 0 }  ][line width=0.75]    (10.93,-3.29) .. controls (6.95,-1.4) and (3.31,-0.3) .. (0,0) .. controls (3.31,0.3) and (6.95,1.4) .. (10.93,3.29)   ;
\draw [color={rgb, 255:red, 0; green, 0; blue, 0 }  ,draw opacity=1 ]   (353.02,283.62) -- (512.13,641.54) ;
\draw [shift={(512.94,643.37)}, rotate = 246.03] [color={rgb, 255:red, 0; green, 0; blue, 0 }  ,draw opacity=1 ][line width=0.75]    (10.93,-3.29) .. controls (6.95,-1.4) and (3.31,-0.3) .. (0,0) .. controls (3.31,0.3) and (6.95,1.4) .. (10.93,3.29)   ;
\draw [color={rgb, 255:red, 128; green, 128; blue, 128 }  ,draw opacity=1 ]   (581.3,120.02) -- (374.37,489.63) ;
\draw [shift={(373.4,491.38)}, rotate = 299.24] [color={rgb, 255:red, 128; green, 128; blue, 128 }  ,draw opacity=1 ][line width=0.75]    (10.93,-3.29) .. controls (6.95,-1.4) and (3.31,-0.3) .. (0,0) .. controls (3.31,0.3) and (6.95,1.4) .. (10.93,3.29)   ;
\draw    (581.3,120.02) -- (738.2,481.63) ;
\draw [shift={(739,483.46)}, rotate = 246.54] [color={rgb, 255:red, 0; green, 0; blue, 0 }  ][line width=0.75]    (10.93,-3.29) .. controls (6.95,-1.4) and (3.31,-0.3) .. (0,0) .. controls (3.31,0.3) and (6.95,1.4) .. (10.93,3.29)   ;
\draw [color={rgb, 255:red, 128; green, 128; blue, 128 }  ,draw opacity=1 ]   (581.3,120.02) -- (533.55,547.41) ;
\draw [shift={(533.32,549.4)}, rotate = 276.38] [color={rgb, 255:red, 128; green, 128; blue, 128 }  ,draw opacity=1 ][line width=0.75]    (10.93,-3.29) .. controls (6.95,-1.4) and (3.31,-0.3) .. (0,0) .. controls (3.31,0.3) and (6.95,1.4) .. (10.93,3.29)   ;
\draw [color={rgb, 255:red, 208; green, 2; blue, 27 }  ,draw opacity=1 ][line width=1.5]    (353.02,283.62) -- (580.4,121.09) ;
\draw [color={rgb, 255:red, 230; green, 219; blue, 219 }  ,draw opacity=1 ]   (176.06,596.32) -- (310.4,658.98) ;
\draw [color={rgb, 255:red, 230; green, 219; blue, 219 }  ,draw opacity=1 ]   (310.4,658.98) -- (495.91,607.92) ;
\draw [color={rgb, 255:red, 230; green, 219; blue, 219 }  ,draw opacity=1 ]   (176.06,596.32) -- (495.91,607.92) ;
\draw [color={rgb, 255:red, 230; green, 219; blue, 219 }  ,draw opacity=1 ]   (177.1,596.95) -- (405.38,433.35) ;
\draw    (353.06,284.75) -- (305.31,712.14) ;
\draw [shift={(305.08,714.12)}, rotate = 276.38] [color={rgb, 255:red, 0; green, 0; blue, 0 }  ][line width=0.75]    (10.93,-3.29) .. controls (6.95,-1.4) and (3.31,-0.3) .. (0,0) .. controls (3.31,0.3) and (6.95,1.4) .. (10.93,3.29)   ;
\draw [color={rgb, 255:red, 230; green, 219; blue, 219 }  ,draw opacity=1 ]   (494.55,611.95) -- (722.71,448.54) ;

\draw (34.6,25.78) node [anchor=north west][inner sep=0.75pt]   [align=left] {$\displaystyle C$};
\draw (145.54,75.76) node [anchor=north west][inner sep=0.75pt]   [align=left] {$\displaystyle Q$};
\draw (50.46,135.19) node [anchor=north west][inner sep=0.75pt]   [align=left] {$\displaystyle E$};
\draw (341.2,244.61) node [anchor=north west][inner sep=0.75pt]  [font=\large] [align=left] {$\displaystyle a_{0}$};
\draw (567.91,90.96) node [anchor=north west][inner sep=0.75pt]  [font=\large] [align=left] {$\displaystyle a\ _{\log q}$};
\draw (165.01,545.99) node [anchor=north west][inner sep=0.75pt]  [rotate=-297.34] [align=left] {DC $\displaystyle ( 2,-1,1)$};
\draw (434.92,491.48) node [anchor=north west][inner sep=0.75pt]  [rotate=-70.73] [align=left] {TP $\displaystyle ( -2,1,1)$};
\draw (442.87,312.66) node [anchor=north west][inner sep=0.75pt]  [rotate=-297.34] [align=left] {DC $\displaystyle ( 2,-1,1)$};
\draw (527.55,360.76) node [anchor=north west][inner sep=0.75pt]  [rotate=-276.16] [align=left] {RC (0,-1,1)};
\draw (657.56,230.93) node [anchor=north west][inner sep=0.75pt]  [rotate=-65.6] [align=left] {TP $\displaystyle ( -2,1,1)$};
\draw (445.14,186.55) node [anchor=north west][inner sep=0.75pt]  [rotate=-321.16] [align=left] {\textcolor[rgb]{0.82,0.01,0.11}{open}};
\draw (293.43,572.17) node [anchor=north west][inner sep=0.75pt]  [rotate=-276.16] [align=left] {RC (0,-1,1)};

\end{tikzpicture}

%% file: figures/attanability_equal_to_half.tex
\tikzset{every picture/.style={line width=0.75pt}} 

\begin{tikzpicture}[x=0.75pt,y=0.75pt,yscale=-.5,xscale=.5]
\scriptsize
\draw    (173.82,152.69) -- (154.98,75.92) ;
\draw [shift={(154.5,73.98)}, rotate = 76.21] [color={rgb, 255:red, 0; green, 0; blue, 0 }  ][line width=0.75]    (10.93,-3.29) .. controls (6.95,-1.4) and (3.31,-0.3) .. (0,0) .. controls (3.31,0.3) and (6.95,1.4) .. (10.93,3.29)   ;
\draw    (173.82,152.69) -- (232.91,135.49) ;
\draw [shift={(234.83,134.93)}, rotate = 163.77] [color={rgb, 255:red, 0; green, 0; blue, 0 }  ][line width=0.75]    (10.93,-3.29) .. controls (6.95,-1.4) and (3.31,-0.3) .. (0,0) .. controls (3.31,0.3) and (6.95,1.4) .. (10.93,3.29)   ;
\draw    (173.82,152.69) -- (147.4,185.73) ;
\draw [shift={(146.15,187.3)}, rotate = 308.64] [color={rgb, 255:red, 0; green, 0; blue, 0 }  ][line width=0.75]    (10.93,-3.29) .. controls (6.95,-1.4) and (3.31,-0.3) .. (0,0) .. controls (3.31,0.3) and (6.95,1.4) .. (10.93,3.29)   ;
\draw  [color={rgb, 255:red, 148; green, 197; blue, 226 }  ,draw opacity=1 ][fill={rgb, 255:red, 148; green, 197; blue, 226 }  ,fill opacity=1 ] (501.84,264.72) -- (952.34,117.39) -- (909.72,589.26) -- (459.23,736.59) -- cycle ;
\draw  [color={rgb, 255:red, 147; green, 198; blue, 226 }  ,draw opacity=1 ][fill={rgb, 255:red, 147; green, 198; blue, 226 }  ,fill opacity=1 ] (501.16,259.53) -- (456.09,736.19) -- (312.51,657.77) -- cycle ;
\draw  [color={rgb, 255:red, 147; green, 198; blue, 226 }  ,draw opacity=1 ][fill={rgb, 255:red, 147; green, 198; blue, 226 }  ,fill opacity=1 ] (953.57,117.49) -- (1107.8,527.15) -- (909.72,589.26) -- cycle ;
\draw [color={rgb, 255:red, 230; green, 219; blue, 219 }  ,draw opacity=1 ]   (314.28,658.81) -- (457.8,737.65) ;
\draw [color={rgb, 255:red, 230; green, 219; blue, 219 }  ,draw opacity=1 ]   (457.8,737.65) -- (656,673.41) ;
\draw [color={rgb, 255:red, 230; green, 219; blue, 219 }  ,draw opacity=1 ]   (314.28,658.81) -- (656,673.41) ;
\draw [color={rgb, 255:red, 230; green, 219; blue, 219 }  ,draw opacity=1 ]   (766.2,510.43) -- (909.72,589.26) ;
\draw [color={rgb, 255:red, 230; green, 219; blue, 219 }  ,draw opacity=1 ]   (909.72,589.26) -- (1107.91,525.03) ;
\draw [color={rgb, 255:red, 230; green, 219; blue, 219 }  ,draw opacity=1 ]   (766.2,510.43) -- (1107.91,525.03) ;
\draw    (502.23,264.66) -- (280.97,730) ;
\draw [shift={(280.11,731.81)}, rotate = 295.43] [color={rgb, 255:red, 0; green, 0; blue, 0 }  ][line width=0.75]    (10.93,-3.29) .. controls (6.95,-1.4) and (3.31,-0.3) .. (0,0) .. controls (3.31,0.3) and (6.95,1.4) .. (10.93,3.29)   ;
\draw [color={rgb, 255:red, 128; green, 128; blue, 128 }  ,draw opacity=1 ]   (502.23,264.66) -- (672.38,715.34) ;
\draw [shift={(673.08,717.21)}, rotate = 249.32] [color={rgb, 255:red, 128; green, 128; blue, 128 }  ,draw opacity=1 ][line width=0.75]    (10.93,-3.29) .. controls (6.95,-1.4) and (3.31,-0.3) .. (0,0) .. controls (3.31,0.3) and (6.95,1.4) .. (10.93,3.29)   ;
\draw    (502.23,264.66) -- (451.16,802.81) ;
\draw [shift={(450.97,804.8)}, rotate = 275.42] [color={rgb, 255:red, 0; green, 0; blue, 0 }  ][line width=0.75]    (10.93,-3.29) .. controls (6.95,-1.4) and (3.31,-0.3) .. (0,0) .. controls (3.31,0.3) and (6.95,1.4) .. (10.93,3.29)   ;
\draw [color={rgb, 255:red, 128; green, 128; blue, 128 }  ,draw opacity=1 ]   (954.14,116.28) -- (732.89,581.62) ;
\draw [shift={(732.03,583.42)}, rotate = 295.43] [color={rgb, 255:red, 128; green, 128; blue, 128 }  ,draw opacity=1 ][line width=0.75]    (10.93,-3.29) .. controls (6.95,-1.4) and (3.31,-0.3) .. (0,0) .. controls (3.31,0.3) and (6.95,1.4) .. (10.93,3.29)   ;
\draw    (954.14,116.28) -- (1124.29,566.95) ;
\draw [shift={(1125,568.82)}, rotate = 249.32] [color={rgb, 255:red, 0; green, 0; blue, 0 }  ][line width=0.75]    (10.93,-3.29) .. controls (6.95,-1.4) and (3.31,-0.3) .. (0,0) .. controls (3.31,0.3) and (6.95,1.4) .. (10.93,3.29)   ;
\draw [color={rgb, 255:red, 128; green, 128; blue, 128 }  ,draw opacity=1 ][fill={rgb, 255:red, 155; green, 155; blue, 155 }  ,fill opacity=1 ]   (954.14,116.28) -- (903.07,654.42) ;
\draw [shift={(902.88,656.41)}, rotate = 275.42] [color={rgb, 255:red, 128; green, 128; blue, 128 }  ,draw opacity=1 ][line width=0.75]    (10.93,-3.29) .. controls (6.95,-1.4) and (3.31,-0.3) .. (0,0) .. controls (3.31,0.3) and (6.95,1.4) .. (10.93,3.29)   ;
\draw [color={rgb, 255:red, 230; green, 219; blue, 219 }  ,draw opacity=1 ]   (457.8,737.65) -- (1109,526) ;
\draw [color={rgb, 255:red, 230; green, 219; blue, 219 }  ,draw opacity=1 ]   (314.28,658.81) -- (766.2,510.43) ;
\draw [color={rgb, 255:red, 208; green, 2; blue, 27 }  ,draw opacity=1 ][line width=1.5]    (501.77,263.75) -- (954.14,116.28) ;

\draw (131.6,77.78) node [anchor=north west][inner sep=0.75pt]   [align=left] {$\displaystyle C$};
\draw (242.54,127.76) node [anchor=north west][inner sep=0.75pt]   [align=left] {$\displaystyle Q$};
\draw (147.46,187.19) node [anchor=north west][inner sep=0.75pt]   [align=left] {$\displaystyle E$};
\draw (490.28,218.55) node [anchor=north west][inner sep=0.75pt]  [font=\large] [align=left] {$\displaystyle a_{0}$};
\draw (941.44,82.69) node [anchor=north west][inner sep=0.75pt]  [font=\large] [align=left] {$\displaystyle a\ _{\log q}$};
\draw (697.75,174.15) node [anchor=north west][inner sep=0.75pt]  [rotate=-340.11] [align=left] {\textcolor[rgb]{0.82,0.01,0.11}{open}};
\draw (346.01,506.99) node [anchor=north west][inner sep=0.75pt]  [rotate=-297.34] [align=left] {DC $\displaystyle ( 2,-1,1)$};
\draw (445.43,544.17) node [anchor=north west][inner sep=0.75pt]  [rotate=-276.16] [align=left] {RC (0,-1,1)};
\draw (571.92,464.48) node [anchor=north west][inner sep=0.75pt]  [rotate=-70.73] [align=left] {TP $\displaystyle ( -2,1,1)$};
\draw (792.01,389.99) node [anchor=north west][inner sep=0.75pt]  [rotate=-297.34] [align=left] {DC $\displaystyle ( 2,-1,1)$};
\draw (898.43,431.17) node [anchor=north west][inner sep=0.75pt]  [rotate=-276.16] [align=left] {RC (0,-1,1)};
\draw (1063.92,316.48) node [anchor=north west][inner sep=0.75pt]  [rotate=-70.73] [align=left] {TP $\displaystyle ( -2,1,1)$};

\end{tikzpicture}